\documentclass[a4paper,11pt]{article}
\usepackage[margin=1in]{geometry}
\usepackage{amsmath,amsfonts, amsthm, mathtools, amssymb}
\usepackage{bm}
\usepackage{bbm}
\usepackage{algorithm}
\usepackage{algpseudocode}
\usepackage{silence}
\usepackage{booktabs}
\usepackage{makecell}
\usepackage[hang,small,bf]{caption}
\usepackage[subrefformat=parens]{subcaption}
\usepackage{graphicx}
\usepackage{thm-restate}

\DeclarePairedDelimiter{\pare}{(}{)}
\DeclarePairedDelimiter{\set}{\{}{\}}
\DeclarePairedDelimiter{\brac}{[}{]}
\DeclarePairedDelimiter{\nor}{\|}{\|}
\DeclarePairedDelimiter{\abs}{|}{|}
\DeclarePairedDelimiter{\ceil}{\lceil}{\rceil}
\DeclarePairedDelimiter{\floor}{\lfloor}{\rfloor}

\newcommand{\OPT}{\mathrm{OPT}}
\newcommand{\ROPT}{R_\mathrm{OPT}}
\newcommand{\LOPT}{L_\mathrm{OPT}}
\newcommand{\NC}{\mathrm{NC}}
\newcommand{\E}{\mathbb{E}}
\newcommand{\mE}{\mathrm{E}}
\newcommand{\cP}{\mathcal{P}}

\newcommand{\vol}{\mathrm{vol}}
\newcommand{\bbm}{\mathbbm{1}}

\theoremstyle{plain}
\newtheorem{dfn}{Definition}[section]

\newtheorem{lem}[dfn]{Lemma}
\newtheorem{clm}[dfn]{Claim}
\newtheorem{thm}[dfn]{Theorem}
\newtheorem{cor}[dfn]{Corollary}
\newtheorem{rem}[dfn]{Remark}

\newtheorem{prob}[dfn]{Problem}

\begin{document}
\title{{\bf EPTAS for Hard Graph Cut Problems for Dense Graphs}}
\author{Kaisei Deguchi\thanks{The University of Tokyo, Tokyo, Japan, \texttt{deguchi-kaisei919@g.ecc.u-tokyo.ac.jp}.\\ 
    Supported by JSPS JP25K24465, 26K21777 and JST ASPIRE JPMJAP2302.}
    \and 
    Ken-ichi Kawarabayashi\thanks{National Institute of Informatics \& The University of Tokyo, Tokyo, \texttt{k\_keniti@nii.ac.jp}.\\ 
    Supported by JSPS Kakenhi JP25K24465, 26K21777 and by JST ASPIRE JPMJAP2302.}
    \and     
    Hiroaki Mori\thanks{The University of Tokyo, Tokyo, Japan, \texttt{mori-hiroaki0102@g.ecc.u-tokyo.ac.jp}.\\ 
    Supported by JSPS JP25K24465, 26K21777 and JST ASPIRE JPMJAP2302.
    }
}

\date{\today}
\maketitle

\begin{abstract}
  Everywhere-$\delta$-dense graphs are defined as graphs on $n$ vertices in which every vertex has degree at least $\delta n$ for some constant $\delta > 0$.
  Approximation schemes are vital for handling NP-hard optimization problems, but for many graph cut problems, existing PTAS algorithms often suffer from running times of $n^{f(1/\varepsilon)}$.
  In this paper, we bring PTASs down to EPTASs for several fundamental minimization problems on everywhere-$\Omega(1)$-dense graphs.
  Specifically, we present the first Efficient Polynomial-Time Approximation Scheme (EPTAS), running in time $f(1/\varepsilon)n^{O(1)}$, for the \textsc{ConstrainedMinCut} problem under a global constraint on vertex weights, a problem that captures \textsc{BalancedSeparator} and \textsc{SmallSetExpansion}.
  Moreover, we give the first EPTASs for \textsc{MinQuotientCut} and \textsc{ProductSparsestCut} on everywhere-$\delta$-dense graphs with integer-valued dense vertex weights;
  these problems generalize the four well-known problems \textsc{UniformSparsestCut}, \textsc{EdgeExpansion}, \textsc{Conductance}, and \textsc{NormalizedCut}.

  Our main technical contribution is an EPTAS for \textsc{ConstrainedMinCut}, based on the weak regularity lemma and sampling and estimation techniques.
  We then obtain EPTASs for \textsc{MinQuotientCut} and \textsc{ProductSparsestCut} via a unified reduction that invokes this algorithm as a subroutine.
  In contrast, previous works giving PTASs for these problems on everywhere-$\delta$-dense graphs typically rely on powerful tools such as the Lasserre hierarchy or specific integer programming techniques \cite{ARORA1999193,guruswami2011lasserre}, which we avoid.
\end{abstract}

\newpage

\section{Introduction}

\subsection{Approximation Schemes}

It is known that many optimization problems cannot be solved in polynomial time, under the assumption that P $\ne$ NP.
Approximation algorithms are among the most important approaches for handling such hard problems.
To be more precise, approximation algorithms are algorithms that, in polynomial time, produce a solution that is guaranteed to be close to the optimal solution within a particular factor.

In this paper, we are particularly interested in approximation schemes.
Let $f$ be a function that depends only on $\varepsilon$.
A \emph{polynomial-time approximation scheme (PTAS)} is a deterministic algorithm that takes an instance of an optimization problem and a parameter $\varepsilon > 0$ and, in $n^{f(1/\varepsilon)}$ time, produces a solution that is within a factor of $1 + \varepsilon$ of the optimal solution for minimization problems (or within a factor of $1 - \varepsilon$ for maximization problems).
A PTAS allows us to obtain a solution arbitrarily close to optimal by choosing a sufficiently small $\varepsilon$, although the running time may increase rapidly as $\varepsilon$ decreases.
To overcome this difficulty, let us consider an \emph{efficient polynomial-time approximation scheme (EPTAS)}, which is a PTAS with running time $f(1/\varepsilon) n^{O(1)}$, allowing us to obtain a solution arbitrarily close to optimal more efficiently than a PTAS.
The randomized analogue of an EPTAS is called an EPRAS, i.e., an algorithm that runs in time $f(1/\varepsilon)n^{O(1)}$ and that returns a $(1+\varepsilon)$-approximate solution for minimization problems (or a $(1-\varepsilon)$-approximate solution for maximization problems) with constant probability.

Unfortunately, many NP-hard problems do not even admit a PTAS unless P $=$ NP \cite{arora1998proof}. However, this result does not refute the existence of a PTAS for some restricted classes of instances.
For example, Arora, Karger, and Karpinski \cite{ARORA1999193} presented the first PTAS in $n^{O(1/\varepsilon^2)}$ time for several NP-hard problems for dense graphs.
Indeed, there has been extensive work on PTASs for maximization problems for dense graphs \cite{de2000polynomial,frieze1999quick,goldreich1998property,karpinski2001polynomial,mathieu2008yet,yoshida2014approximation}.
In addition, Goldreich, Goldwasser, and Ron presented an EPRAS for \textsc{MaxCut} and additive-error approximation algorithms for some cut problems by applying property testing techniques \cite{goldreich1998property}.
As stated above, our focus is on NP-hard problems for dense graphs.

\subsection{Minimization Problems for Dense Graphs}
To the best of our knowledge, there are few known approximation results for NP-hard minimization problems on dense graphs. Indeed, the density required for PTASs for minimization problems is, in a sense, often stronger than that for maximization problems.
One reason is that maximization problems on dense graphs also have large optimal solution values, which allow large additive error (i.e., small relative error), thereby yielding a PTAS or even an EPTAS.
On the other hand, since density alone in the minimization setting does not, in general, ensure a large optimal solution value, even when the given graph instances are dense, it is difficult to apply techniques from maximization problems to minimization problems directly.
This difference makes it substantially more challenging to design approximation algorithms that achieve small relative error for minimization problems.

Nevertheless, some PTASs for minimization problems for dense graphs are known.
For example, several PTASs for CSPs on dense instances have been developed:
Bazgan, Fernandez de la Vega, and Karpinski gave PTASs for specific classes of minimum CSPs \cite{bazgan2003polynomial,bazgan2002approximability},
Karpinski and Schudy gave EPRASs for minimum CSPs on specific (\textit{fragile} or \textit{rigid}, i.e., if we independently round or guess the assignments of vertices, the objective value only shifts locally) classes of dense instances \cite{karpinski2009linear}, and
Giotis and Guruswami gave a PTAS for correlation clustering with a fixed number of clusters \cite{giotis2005correlation}.
Guruswami and Sinop \cite{guruswami2011lasserre} presented a unified approximation scheme for optimization problems that are formulated as specific quadratic integer programs with positive semidefinite objectives and linear constraints.
Specifically, they gave a $\tfrac{1 + \varepsilon}{\min\set{1, \lambda_r}}$-approximation algorithm in $n^{O(r/\varepsilon^2)}$ time for several graph cut problems, including \textsc{UniformSparsestCut}, \textsc{EdgeExpansion}, \textsc{BalancedSeparator}, \textsc{Conductance}, \textsc{NormalizedCut}, and \textsc{SmallSetExpansion} (for the precise definitions of these problems, see below) by using the Lasserre hierarchy, where $\lambda_r$ is the $r$-th smallest eigenvalue of the normalized Laplacian of the graph.
Their running time $n^{O(r/\varepsilon^2)}$ depends on the $O(r/\varepsilon^2)$ rounds of the Lasserre hierarchy.
If the graph has the property known as low threshold rank, then these algorithms imply PTASs for these problems.

One of the most powerful tools for designing approximation algorithms for dense graphs is the regularity lemma.
Szemerédi's regularity lemma is a fundamental result in graph theory that states that every sufficiently large graph can be approximated by a union of a bounded number of random-like bipartite graphs called regular pairs.
Building on this perspective, Frieze and Kannan \cite{frieze1999quick} presented an algorithmic version of the regularity lemma and gave a constant-time approximation scheme.
Subsequently, many researchers have improved the error bounds of the partition \cite{oveis2013new}, and a deterministic version of the regularity lemma was presented \cite{fox2019fast}.

In this paper, we focus on graph cut problems, including those with global constraints and those minimizing cut ratios.
We obtain an EPTAS for \textsc{ConstrainedMinCut} on the so-called everywhere-$\delta$-dense graphs by using the weak regularity lemma \cite{frieze1999quick,fox2019fast} and sampling and estimation techniques, while avoiding heavier tools such as the Lasserre hierarchy used in previous PTAS results.
Moreover, we obtain EPTASs for \textsc{MinQuotientCut} and \textsc{ProductSparsestCut} through a unified reduction from the EPTAS for \textsc{ConstrainedMinCut}.

More details are given below.

\subsection{Graph Cut Problems}
For two sets $S, T$, let $\mE(S, T)$ denote the set of edges with one endpoint in $S$ and the other in $T$, i.e., $\mE(S, T) = \set{(u, v) \in E \mid u \in S, v \in T}$.
Note that an edge with both endpoints in $S \cap T$ is counted twice.
For $S \subseteq V$, let $\overline{S} := V \setminus S$.
In particular, we call $\mE(S, \overline{S})$ the cut induced by $S$, and $\abs{\mE(S, \overline{S})}$ the cut size of $S$.
Let $d_v$ be the degree of vertex $v \in V$.
We define the volume of $S \subseteq V$ as $\vol(S) := \sum_{v \in S} d_v$.
For $S \subseteq V$ and a function $c$ on $V$, let $c(S)$ denote $\sum_{v \in S} c_v$.

We now define several graph cut problems.

\begin{prob}[\textsc{BalancedSeparator}]
  Given an unweighted graph $G=(V, E)$, a parameter $t \ (1 \le t \le n/2)$, and a constant $\zeta \ (\zeta > 0)$,
  find a cut $S$ such that $t(1 - \zeta) \le |S| \le t(1 + \zeta)$ and $\abs{\mE(S, \overline{S})}$ is minimized.
\end{prob}
When $t = n/2$ and $\zeta = 0$, the problem coincides with the \textsc{MinBisection} problem, although this case is outside the scope of the above definition.

\begin{prob}[\textsc{UniformSparsestCut}]
  Given an unweighted graph $G=(V, E)$, find a cut $S$ such that the
  sparsity for uniform demand $\Phi(S) = \frac{\abs{\mE(S, \overline{S})}} {\abs{S}\abs{\overline{S}}}$ is minimized.
  We define $\Phi(G) := \min_{S \subseteq V} \Phi(S)$.
\end{prob}

\begin{prob}[\textsc{EdgeExpansion}]
  Given an unweighted graph $G=(V, E)$, find a cut $S$ such that the
  edge expansion $h(S) = \frac{\abs{\mE(S, \overline{S})}} {\min(|S|, |\overline{S}|)}$ is minimized.
  We define $h(G) := \min_{S \subseteq V} h(S)$.
\end{prob}

\begin{prob}[\textsc{SmallSetExpansion}]
  Given an unweighted graph $G=(V, E)$, a parameter $\rho \ (1 \le \rho \le \vol(V)/2)$, and a constant $\zeta \ (\zeta > 0)$, find a cut $S$ such that $\rho(1-\zeta) \le \vol(S) \le \rho(1+\zeta)$ and $\abs{\mE(S, \overline{S})}$ is minimized.
\end{prob}

\begin{prob}[\textsc{Conductance}]
  Given an unweighted graph $G=(V, E)$, find a cut $S$ such that the
  conductance $\phi(S) = \frac{\abs{\mE(S, \overline{S})}} {\min(\vol(S), \vol(\overline{S}))}$ is minimized.
  We define $\phi(G) := \min_{S \subseteq V} \phi(S)$.
\end{prob}

\begin{prob}[\textsc{NormalizedCut}]
  Given an unweighted graph $G=(V, E)$, find a cut $S$ such that the
  normalized cut $\NC(S) = \frac{\abs{\mE(S, \overline{S})}} {\vol(S) \vol(\overline{S})}$ is minimized.
  We define $\NC(G) := \min_{S \subseteq V} \NC(S)$.
\end{prob}

\begin{prob}[\textsc{ConstrainedMinCut}]
  Given an unweighted graph $G=(V, E)$ with an integer vertex-weight function $c: V \rightarrow \mathbb{Z}_{\ge 0}$ such that $c_v \le n$ for every $v \in V$,
  a parameter $\rho \ (1 \le \rho \le c(V)/2)$, and a constant $\zeta \ (\zeta > 0)$,
  find a cut $S$ such that $\rho(1-\zeta) \le c(S) \le \rho(1+\zeta)$ and $\abs{\mE(S, \overline{S})}$ is minimized.
  We refer to the constraint $\rho(1 - \zeta) \le c(S) \le \rho(1 + \zeta)$ as the cost constraint.
\end{prob}

\begin{prob}[\textsc{MinQuotientCut}]
  Given an unweighted graph $G=(V, E)$ with an integer vertex-weight function $c: V \rightarrow \mathbb{Z}_{\ge 0}$ such that $c_v \le n$ for every $v \in V$,
  find a cut $S$ such that the
  quotient cut $q(S) = \frac{\abs{\mE(S, \overline{S})}} {\min(c(S), c(\overline{S}))}$ is minimized.
  We define $q(G) := \min_{S \subseteq V} q(S)$.
\end{prob}

\begin{prob}[\textsc{ProductSparsestCut}]
  Given an unweighted graph $G=(V, E)$ with an integer vertex-weight function $c: V \rightarrow \mathbb{Z}_{\ge 0}$ such that $c_v \le n$ for every $v \in V$,
  find a cut $S$ such that the
  sparsity for product demand $\Phi^\times(S) = \frac{\abs{\mE(S, \overline{S})}} {c(S) c(\overline{S})}$ is minimized.
  We define $\Phi^\times(G) := \min_{S \subseteq V} \Phi^\times(S)$.
\end{prob}

The \textsc{SmallSetExpansion} problem \cite{raghavendra2010graph} plays an essential role in complexity theory and hardness of approximation.
Under the Small Set Expansion Hypothesis (SSEH), strong inapproximability results are known for many problems \cite{austrin2012inapproximability,gandhi2014set,manurangsi2018inapproximability}.
Moreover, Raghavendra and Steurer \cite{raghavendra2010graph} showed that SSEH implies the Unique Games Conjecture (UGC), and UGC yields tight (often optimal) inapproximability results for a wide range of problems \cite{khot2007optimal, khot2008vertex, raghavendra2008optimal}.

\subsection{Our Contributions}

Below, we present our main results. We use $\tilde{O}(\cdot)$ to hide polylogarithmic factors.

\begin{dfn}[everywhere-$\delta$-dense graph]
  An unweighted graph $G=(V, E)$ is everywhere-$\delta$-dense if each vertex has degree at least $\delta n$.
\end{dfn}

\begin{dfn}[dense instance of \textsc{ConstrainedMinCut}]
  A \textsc{ConstrainedMinCut} instance $\mathcal{I} = (G, c, \rho, \zeta)$ is called dense if $G$ is an everywhere-$\Omega(1)$-dense graph and there exists a constant $c_0 > 0$ such that $c_0 n \le c_v$ for every vertex $v \in V$.
\end{dfn}

Throughout this paper, we treat the density parameters $\delta$ and $c_0$ as fixed positive constants.
Accordingly, the hidden constants may depend on $\delta$ and $c_0$.
We also treat the interval parameter $\zeta$ as a constant; however, since the EPTAS for \textsc{ConstrainedMinCut} is used as a subroutine in later reductions, we make its dependence on $\zeta$ explicit whenever it is relevant.

\paragraph{Algorithmic Improvements:} Our main result is an EPTAS for \textsc{ConstrainedMinCut}, which yields EPTASs for several problems on everywhere-$\delta$-dense graphs, either by direct implication or by a short unified reduction.
This improves upon prior PTAS results, whose running times are of the form $n^{f(1/\varepsilon)}$.

Table~\ref{tab:prior} summarizes prior results and our contributions.
The heavy $2^{\tilde{O}(1/\varepsilon^c)}$ dependence in the running time stems from our use of the deterministic weak regularity lemma \cite{fox2019fast}, which yields a decomposition of width $O(\varepsilon^{-16})$.
If we allow randomization, we may instead use the randomized weak regularity lemma \cite{frieze1999quick} with width $O(\varepsilon^{-2})$, thereby improving the exponent of $1/\varepsilon$ in the $2^{\tilde{O}(\cdot)}$ term (with high probability).

\begin{restatable}[Main]{thm}{maintheorem}
  \label{thm:eptas_con}
  There exists an EPTAS for the
  \textsc{ConstrainedMinCut} problem
  on dense instances.
\end{restatable}
In particular, by setting $\rho = tn$ and $c_{v}= n$ for every vertex $v \in V$, we may take $c_0 = 1$, and the cost constraint coincides with the cardinality constraint in the instance \textsc{BalancedSeparator}$(G, t, \zeta)$.
Hence, Theorem~\ref{thm:eptas_con} yields an EPTAS for \textsc{BalancedSeparator}$(G, t, \zeta)$.
Similarly, by setting $c_{v}= d_{v}$ for every vertex $v \in V$, the cost constraint coincides with the volume constraint in \textsc{SmallSetExpansion};
moreover, since $G$ is everywhere-$\delta$-dense, we may take $c_0 = \delta$.
Thus, Theorem~\ref{thm:eptas_con} also yields an EPTAS for \textsc{SmallSetExpansion} and hence the following corollary holds.
\begin{cor}
  \label{cor:bs_sse_eptas}
  There exist EPTASs for the problems
  \textsc{BalancedSeparator} and \textsc{SmallSetExpansion}
  on everywhere-$\delta$-dense graphs, and their running times are $\varepsilon^{-O(1)}n^{O(1)} + 2^{\tilde{O}(1/(\zeta^{16}\varepsilon^{32}))}$.
\end{cor}

Moreover, dense instances of \textsc{MinQuotientCut} and \textsc{ProductSparsestCut} admit a short reduction to \textsc{ConstrainedMinCut}.
\begin{dfn}[dense instance of \textsc{MinQuotientCut} and \textsc{ProductSparsestCut}]
  A \textsc{MinQuotientCut} instance $\mathcal{I} = (G, c)$ is called dense if $G$ is an everywhere-$\Omega(1)$-dense graph and there exists a constant $c_0 > 0$ such that $c_0 n \le c_v$ for every vertex $v \in V$.
\end{dfn}
Therefore, the EPTAS for \textsc{ConstrainedMinCut} yields EPTASs for dense instances of \textsc{MinQuotientCut} and \textsc{ProductSparsestCut}.
\begin{restatable}{thm}{reduction}
  \label{thm:mqc_psc}
  There exist EPTASs for the problems
  \textsc{MinQuotientCut} and \textsc{ProductSparsestCut}
  on dense instances, and their running times are $\varepsilon^{-O(1)}n^{O(1)} + 2^{\tilde{O}(1/\varepsilon^{48})}$.
\end{restatable}

\textsc{EdgeExpansion} and \textsc{Conductance} (resp. \textsc{UniformSparsestCut} and \textsc{NormalizedCut}) are special cases of \textsc{MinQuotientCut} (resp. \textsc{ProductSparsestCut}).
Hence, we have the following corollary.

\begin{restatable}{cor}{corollaryred}
  \label{cor:cor4}
  There exist EPTASs for the problems
  \textsc{UniformSparsestCut}, \textsc{EdgeExpansion}, \textsc{Conductance}, and \textsc{NormalizedCut}
  on everywhere-$\delta$-dense graphs, and their running times are $\varepsilon^{-O(1)}n^{O(1)} + 2^{\tilde{O}(1/\varepsilon^{48})}$.
\end{restatable}

Previous works giving PTASs for these problems on everywhere-$\delta$-dense graphs typically rely on powerful tools such as the Lasserre hierarchy \cite{guruswami2011lasserre} or specific integer programming techniques \cite{ARORA1999193}.
As a consequence, the resulting running time is of the form $n^{f(1/\varepsilon)}$, which yields only PTASs but not EPTASs.
In contrast, using the weak regularity lemma as the main tool \cite{frieze1999quick,fox2019fast}, we obtain EPTASs running in $f(1/\varepsilon)n^{O(1)}$.

\paragraph{Hardness:} Unlike problems such as \textsc{MaxCut}, for \textsc{ConstrainedMinCut}, the assumption that the input graph is everywhere-$\delta$-dense cannot be weakened to the condition $\abs{E} \ge \Omega(n^2)$, often used to define so-called dense graphs.
To see this, we first observe that by an argument similar to that in \cite{ARORA1999193}, if there exists a PTAS for \textsc{BalancedSeparator} on dense graphs, then there exists a PTAS for \textsc{BalancedSeparator} on general graphs, which clearly contradicts the known hardness result.

Given any (not necessarily dense) instance of \textsc{BalancedSeparator} $(G = (V, E), t (= \Theta(n)), \zeta (= \Theta(1)))$ with $n$ vertices,
we construct an instance $G'$ for \textsc{BalancedSeparator} with $\abs{E(G')} \ge \Omega(\abs{V(G')}^2)$ and show that we can recover a $(1 + \varepsilon)$-approximate solution to the original \textsc{BalancedSeparator} instance from a nearly optimal solution to the constructed instance.

We add a clique $K$ with $t(1+\zeta)+n$ vertices to $G$, and let $G'$ denote the resulting graph, which is dense.
We show that if we run the PTAS on $(G', t, \zeta)$, we can recover a $(1+\varepsilon)$-approximate solution to the original \textsc{BalancedSeparator} instance.
Let $S$ be a feasible solution to the constructed instance, where $X = S \cap V$ and $Y = S \cap K$.

Now, our goal is to show that we can recover a cut $T$ from $S$ such that $\abs{\mE_{G}(T, \overline{T})} \le \abs{\mE_{G'}(S, \overline{S})}$ holds.
The cut size of $S$ is
\begin{equation*}
  \abs{\mE_{G'}(S, \overline{S})} = \abs{\mE_G(X, \overline{X})} + \abs{Y}(\abs{K}-\abs{Y}).
\end{equation*}
When $X$ is a feasible solution to the original instance, the condition $\abs{\mE_{G}(T, \overline{T})} \le \abs{\mE_{G'}(S, \overline{S})}$ holds.

When $\abs{X} < t(1-\zeta)$ holds, we construct a feasible cut $T = X\cup A$ where $A \subseteq V\setminus X$ is an arbitrary vertex set with $\abs{X \cup A} = t(1-\zeta)$.
By the feasibility of $S$, $\abs{A}$ is at most $\abs{Y}$.
Since $\abs{K}-\abs{Y} = t(1+\zeta) + n - \abs{Y} \ge n$ holds,
$\abs{\mE_G(T, \overline{T})} \le \abs{\mE_G(X, \overline{X})} + \abs{A}n \le  \abs{\mE_G(X, \overline{X})} + \abs{Y}n \le \abs{\mE_G(X, \overline{X})} + \abs{Y}(\abs{K}-\abs{Y}) = \abs{\mE_{G'}(S, \overline{S})}$.
The above argument shows that every feasible solution $S$ of $G'$ can be converted
into a feasible solution $T$ of $G$ such that
\begin{equation*}
  \abs{\mE_G(T,\overline T)} \le \abs{\mE_{G'}(S,\overline S)}.
\end{equation*}
Hence $\OPT_G\le \OPT_{G'}$.
Conversely, every feasible solution of $G$ remains feasible for $G'$ when we take no vertices from $K$.
Therefore, $\OPT_{G'}\le \OPT_G$, and hence $\OPT_G = \OPT_{G'}$ holds.

When we run the PTAS for the constructed instance and obtain a $(1+\varepsilon)$-approximate solution $S$, we can construct $T$ such that $\abs{\mE_{G}(T, \overline{T})} \le \abs{\mE_{G'}(S, \overline{S})} \le (1+\varepsilon)\OPT_{G'} = (1+\varepsilon)\OPT_{G}$ holds.
Therefore, if there exists a PTAS for \textsc{BalancedSeparator} on dense graphs, it yields a PTAS for \textsc{BalancedSeparator} on general graphs.

Thus, it is unlikely that there exist PTASs for either \textsc{BalancedSeparator} or \textsc{ConstrainedMinCut} on dense graphs.

\paragraph{Comparison with prior PTASs:}
Our algorithm builds on the same sampling and estimation paradigm as prior PTASs for dense graph problems.
Most prior PTASs for dense CSPs \cite{bazgan2003polynomial,bazgan2002approximability,giotis2005correlation} treat unconstrained problems, and a few \cite{karpinski2009linear} handle hard \emph{local} constraints, such as those in the \textsc{MultiwayCut} problem.
On the other hand, \textsc{BalancedSeparator} and \textsc{SmallSetExpansion} require global interval cardinality or volume constraints, and hence these problems involve non-fragile global constraints.
Therefore, these previous works do not provide PTASs for problems that require non-fragile global constraints.
We give more details in the next subsection.

The most relevant prior work is the PTAS for \textsc{BalancedSeparator} on dense graphs due to Arora, Karger, and Karpinski \cite{ARORA1999193}, which runs in $n^{O(1/\varepsilon^3)}$ time and is therefore not an EPTAS.
Our results improve upon this in two directions: we obtain an EPTAS running in $f(1/\varepsilon)n^{O(1)}$ time, and we extend the framework to the more general \textsc{ConstrainedMinCut} problem, which captures both \textsc{BalancedSeparator} and \textsc{SmallSetExpansion}.
A naive sampling and estimation framework does not produce solutions satisfying the global constraint, and we introduce an additional rebalancing step that adjusts the cut to satisfy the constraint without significantly increasing the cut size.

\begin{table}[t]
  \centering
  \scriptsize
  \setlength{\tabcolsep}{1pt}
  \begin{tabular}{@{} |c|c|c|c|c|c| @{}}
    \toprule
    Problems & Reference                                                                        & Guarantee & Time & Type & Graph class \\

    \midrule
    \makecell[c]{\textsc{UniformSparsestCut}                                                                                            \\
    \textsc{EdgeExpansion}                                                                                                              \\
    \textsc{BalancedSeparator}                                                                                                          \\
    \textsc{Conductance}                                                                                                                \\
    \textsc{NormalizedCut}                                                                                                              \\
      \textsc{SmallSetExpansion}}
             & \cite{guruswami2011lasserre}
             & $\tfrac{1 + \varepsilon}{\min\set{1, \lambda_r}}$
             & $n^{O(r/\varepsilon^2)}$
             & \makecell[c]{PRAS                                                                                                        \\(if $^\exists r = O(1)$, \\$\lambda_r \ge 1 -\varepsilon$)}
             & \makecell[c]{any                                                                                                         \\(where $\lambda_r$ denotes the $r$-th smallest eigenvalue \\of the normalized Laplacian)}\\

    \midrule
    \textsc{BalancedSeparator}
             & \cite{ARORA1999193}
             & $1 + \varepsilon$
             & $n^{O(1/\varepsilon^3)}$
             & PTAS
             & everywhere-$\delta$-dense                                                                                                \\

    \midrule
    \textsc{BalancedSeparator}
             & \cite{frieze1996regularity}
             & $1 + \varepsilon$
             & $\tilde{O}(\varepsilon^{-2})n^{O(1)} + 2^{2^{(1/\varepsilon)^{O(1)}}}$
             & EPRAS
             & everywhere-$\delta$-dense                                                                                                \\

    \midrule
    \textsc{UniformSparsestCut}
             & \cite{frieze1999quick}
             & additive $\varepsilon$
             & $n + 2^{\tilde{O}(1/\varepsilon^2)}$
             & -
             & any                                                                                                                      \\

    \midrule
    \makecell[c]{
    \textsc{BalancedSeparator}                                                                                                          \\
      \textsc{SmallSetExpansion}
    }
             & \textbf{Ours}
             & $1 + \varepsilon$
             & $\varepsilon^{-O(1)}n^{O(1)} + 2^{\tilde{O}(1 / (\zeta^{16} \varepsilon^{32}))}$
             & \textbf{EPTAS}
             & everywhere-$\delta$-dense                                                                                                \\

    \midrule
    \makecell[c]{
    \textsc{UniformSparsestCut}                                                                                                         \\
    \textsc{EdgeExpansion}                                                                                                              \\
    \textsc{Conductance}                                                                                                                \\
      \textsc{NormalizedCut}
    }
             & \textbf{Ours}
             & $1 + \varepsilon$
             & $\varepsilon^{-O(1)}n^{O(1)} + 2^{\tilde{O}(1 / \varepsilon^{48})}$
             & \textbf{EPTAS}
             & everywhere-$\delta$-dense                                                                                                \\

    \bottomrule
  \end{tabular}
  \caption{Summary of results.}
  \label{tab:prior}
\end{table}

\subsection{Technical Contributions}
As mentioned above, we leverage weak regularity and specialized sampling and rebalancing procedures for cost-constrained problems.
To this end, technically speaking, our main tool is to generalize the EPRAS for the \textsc{MinBisection} problem in everywhere-$\delta$-dense graphs with small optimal value \cite{ARORA1999193} to an EPRAS for the \textsc{ConstrainedMinCut} problem in the same setting.
For the \textsc{MinBisection} problem, there is a strong structural characterization of the optimal solution, and this characterization is used to design the EPRAS.
\begin{lem}[Lemma 5.1 in \cite{ARORA1999193}]
  Let $(L, R)$ be an optimal cut for the \textsc{MinBisection} problem.
  Then at least one of $L$ and $R$ has the property that every vertex on this side has more than half of its neighbors on the same side.
\end{lem}
Unfortunately, such a strong structural characterization does not extend to the \textsc{ConstrainedMinCut} problem because of the cost constraint.
Indeed, this problem requires an interval cost constraint (this is also the case for \textsc{BalancedSeparator} and \textsc{SmallSetExpansion}).
If, for example, we apply the independent sampling framework in \cite{karpinski2009linear} to the \textsc{MinBisection} problem, we may miss the exact bisection threshold by $O(\varepsilon n)$ vertices.
We may blindly try to ``fix'' the resulting cut by greedily moving $O(\varepsilon n)$ vertices back across the cut to balance the bisection.
But then we may cut up to $O(\varepsilon \delta n^2)$ new edges.
However, when the ``Small OPT'' case happens ($\OPT \ll n^2$), an error of $O(\varepsilon \delta n^2)$ completely violates the $(1+\varepsilon)$ relative approximation.
This is exactly why our Algorithm~\ref{alg:adjust} (the exact knapsack DP) is more than an incremental contribution.
By using an exact DP to carefully resolve the uncertain set $X$ to meet the cost constraint without incurring the $O(\varepsilon \delta n^2)$ greedy penalty, we can successfully overcome this difficulty.

In the end, we show that our sampling procedure correctly classifies vertices with high probability, and that the rebalancing procedure can adjust the cut cost to satisfy the cost constraint without increasing the cut size too much.
This is a key lemma, which enables us to obtain an EPRAS for the \textsc{ConstrainedMinCut} problem on everywhere-$\delta$-dense graphs when the optimal value is small.
We then derandomize the two randomized steps in the algorithms for \textsc{ConstrainedMinCut}.

By using the EPTAS for \textsc{ConstrainedMinCut}, we obtain EPTASs for \textsc{MinQuotientCut} and \textsc{ProductSparsestCut} for dense instances in a unified manner.

\subsection{Organization of This Paper}
In the remainder of the paper, we introduce some notation and the weak regularity lemma in Section~\ref{sec:pre}.
In Section~\ref{sec:main}, for simplicity, we first present an EPRAS for the \textsc{ConstrainedMinCut} problem on everywhere-$\delta$-dense graphs, and then we derandomize the randomized steps of the algorithm.
In Section~\ref{sec:red}, we present EPTASs for \textsc{MinQuotientCut} and \textsc{ProductSparsestCut} on everywhere-$\delta$-dense graphs by using the EPTAS for the \textsc{ConstrainedMinCut} problem as a subroutine.

\section{Preliminaries}
\label{sec:pre}
\subsection{Notation}
In this section, we provide some notation used in this paper.

Let $G=(V, E)$ be an undirected unweighted graph with $n$ vertices.
We denote by $N(v)$ the set of neighbors of vertex $v$.
$\bbm[P]$ is the indicator function that equals 1 if the predicate $P$ is true and 0 otherwise.

\subsection{Weak Regularity Lemma}

We first introduce notation and then state the Frieze-Kannan Weak Regularity Lemma \cite{frieze1999quick}.
\begin{dfn}[Frobenius Norm]
  For an $n \times n$ matrix $A$, its Frobenius norm $\|A\|_{F}$ is defined as
  \begin{equation*}
    \|A\|_{F} = \sqrt{\sum_{i=1}^n \sum_{j=1}^n A_{i,j}^2}.
  \end{equation*}
\end{dfn}

\begin{dfn}[Cut Norm]
  For an $n \times n$ matrix $A$, its cut norm $\|A\|_{C}$ is defined as
  \begin{equation*}
    \|A\|_{C} = \max_{S, T \subseteq [n]} \left| \sum_{i \in S, j \in T} A_{i,j} \right|.
  \end{equation*}
\end{dfn}

\begin{dfn}[Cut Decomposition]
  A cut decomposition of an $n \times n$ matrix $A$ is a decomposition of the form
  \begin{equation*}
    A = d_1\chi_{S_1} \chi_{T_1}^T + d_2\chi_{S_2} \chi_{T_2}^T + \cdots + d_w\chi_{S_w} \chi_{T_w}^T + W
  \end{equation*}
  where for each $t$, $S_t, T_t \subseteq [n]$, $d_t \in \mathbb{R}$, $\chi_{S_t}$ is the indicator vector of set $S_t$, and $W$ is an $n \times n$ matrix called the error matrix.
  We call $w$ the width of the cut decomposition, $\sqrt{\sum_{t=1}^w d_t^2}$ the coefficient length, and $\|W\|_{C}$ the error of the cut decomposition.
\end{dfn}

We are now ready to present the main lemma.
\begin{thm}[Frieze-Kannan Weak Regularity Lemma \cite{frieze1999quick}]
  \label{thm:weak_regularity}
  Suppose $A$ is an $n \times n$ matrix and $\varepsilon, \gamma \in (0, 1)$.
  There exists an algorithm that finds a cut decomposition of $A$ of width $O(\varepsilon^{-2})$, coefficient length at most $\sqrt{27}\|A\|_{F}/n$, and error at most $\varepsilon n \|A\|_{F}$ in time $2^{\tilde{O}(\varepsilon^{-2})}/\gamma^2$ with probability at least $1 - \gamma$.
\end{thm}
\begin{rem}
  For the adjacency matrix of an unweighted graph, we have $\|A\|_F \le n$. Hence, the coefficient length is at most $\sqrt{27}$ and the error is at most $\varepsilon n^2$.
\end{rem}

\section{EPTAS for ConstrainedMinCut}
\label{sec:main}
Throughout this section, we work with dense instances of \textsc{ConstrainedMinCut}.
Since $d_v \le n$ and $c_0 n \le c_v$, $d_v \le \tfrac{1}{c_0}c_v$ holds.

Our first step is to show that when the parameter of the problem is unbalanced, we can find a $(1 + \varepsilon)$-approximate solution for the \textsc{ConstrainedMinCut} problem.

We may assume $\varepsilon \le 1$ without loss of generality.
For any $\zeta > 0$, if $\rho(1 + \zeta) \le c_0 \tfrac{\varepsilon}{2} \delta^2 n^2$, then $c_0 \abs{S}\delta n \le c_0 \vol(S) \le c(S) \le \rho(1+\zeta) \le c_0 \tfrac{\varepsilon}{2} \delta^2 n^2$.
Thus, we have $\abs{S} \le \tfrac{\varepsilon}{2} \delta n$, and hence $\abs{S}^2 \le \abs{S}\tfrac{\varepsilon}{2} \delta n \le \tfrac{\varepsilon}{2} \vol(S)$.
Under this condition, we can conclude the following inequalities hold:
\begin{equation*}
  (1 - \tfrac{\varepsilon}{2})\vol(S) \le \vol(S) -\abs{S}^2 \le \abs{\mE(S, \overline{S})} \le \vol(S).
\end{equation*}
Since $\tfrac{1}{1-\varepsilon/2} \le 1+\varepsilon$ for $\varepsilon \le 1$, this gives $\abs{\mE(S, \overline{S})} \le \vol(S) \le (1+\varepsilon)\abs{\mE(S, \overline{S})}$ for any feasible $S$.
Therefore, for the cut $S'$ that minimizes $\vol(S)$ subject to $c(S) \in [\rho(1-\zeta), \rho(1+\zeta)]$, we have
\begin{equation*}
  \abs{\mE(S', \overline{S'})} \le \vol(S') \le \vol(S^*) \le (1+\varepsilon)\abs{\mE(S^*, \overline{S^*})},
\end{equation*}
where $S^*$ is the optimal cut for \textsc{ConstrainedMinCut}. Hence we obtain a $(1+\varepsilon)$-approximate solution.
This can be done by simple knapsack dynamic programming in $O(n^3)$ time.

Below, we focus on the case where the input parameters are balanced.
Let $\rho' := \rho / c(V)$ be the normalized parameter, and we may assume that $\rho' > \tfrac{c_0 \varepsilon \delta^2}{2(1 + \zeta)}$, since $c(V) \le n^2$.

In \cite{frieze1996regularity}, Frieze and Kannan presented an EPRAS for the \textsc{BalancedSeparator} problem in the large-optimal-value case on everywhere-$\delta$-dense graphs by using the algorithmic regularity lemma.
Combining the result in \cite{ARORA1999193} for the small-optimal-value case, they obtained an EPRAS for the \textsc{BalancedSeparator} problem for everywhere-$\delta$-dense graphs.
In this section, we present an EPRAS for the \textsc{ConstrainedMinCut} problem on everywhere-$\delta$-dense graphs by extending their technique.
We denote the optimal cut by $(\LOPT, \ROPT)$ and its size by $\OPT := \abs{\mE(\LOPT, \ROPT)}$. We divide the algorithm into two cases, based on the size of $\OPT$.

\begin{restatable}{thm}{eprasthm}
  \label{thm:EPRAS}
  There exists an EPRAS for the \textsc{ConstrainedMinCut} problem on dense instances, and its running time is $n^{O(1)} + 2^{\tilde{O}(1 / (\zeta^2 \varepsilon^4))}$.
\end{restatable}

\begin{thm}
  \label{thm:EPTAS}
  There exists an EPTAS for the \textsc{ConstrainedMinCut} problem on dense instances, and its running time is $\varepsilon^{-O(1)}n^{O(1)}+ 2^{\tilde{O}(1 / (\zeta^{16} \varepsilon^{32}))}$.
\end{thm}
The difference in runtime between Theorem~\ref{thm:EPRAS} and Theorem~\ref{thm:EPTAS} is derived from the width of the cut decomposition.

The most important difference from the \textsc{BalancedSeparator} problem is the existence of a cost constraint $c(S) \in [\rho(1-\zeta), \rho(1 + \zeta)]$.

This applies to both the large-optimal-value and small-optimal-value cases of the \textsc{ConstrainedMinCut} problem.
We use the parameter $\alpha$, which will be determined in the proof of the small-optimal-value case.

\subsection{Large OPT Case for ConstrainedMinCut Problem}

If $n \le N(1/\varepsilon)$, the problem can be solved exactly by brute force in $O_\varepsilon(1)$ time.
Thus, in what follows, we assume that $n$ is sufficiently large.

We apply the weak regularity lemma to handle the case in which the optimal cut size $\OPT \ge \alpha n^2$ for the \textsc{ConstrainedMinCut} problem (we refer to this as the ``Large OPT'' case).

Our main result for this case is as follows:

\begin{thm}
  \label{thm:largeOPT}
  For dense instances of \textsc{ConstrainedMinCut} with optimal cut size $\OPT \ge \alpha n^{2}$,
  we can find a cut $\mE(L, \overline{L})$ such that $c(L) \in [\rho(1-\zeta), \rho(1+\zeta)]$ and its cut size is at most $(1 + \varepsilon)\OPT$ in time $2^{\tilde{O}(1/(\alpha\varepsilon)^2)} + n^{O(1)}$ with probability $9/10$.
\end{thm}

Frieze and Kannan \cite{frieze1999quick} applied the weak regularity lemma to obtain additive-error approximation algorithms for many graph problems, including \textsc{MaxCut} and \textsc{UniformSparsestCut}.
In this subsection, building on their framework, we impose a cost constraint to derive an algorithm for the \textsc{ConstrainedMinCut} problem and prove Theorem~\ref{thm:largeOPT}.

The algorithm consists of the following four steps:
\begin{enumerate}
  \item Compute a cut decomposition.
  \item Check feasibility via an LP.
  \item Round down the fractional solution.
  \item Move vertices to satisfy the cost constraint.
\end{enumerate}
Following the algorithm above, we first obtain a cut decomposition with constant width and small error, according to the given precision parameter $\varepsilon$, by applying the weak regularity lemma (Theorem~\ref{thm:weak_regularity}) to the adjacency matrix of the given dense graph.
Then, for each part of the cut decomposition, we discretize its possible contribution to the cut size. For each assignment, we check the feasibility of the assignment by solving a linear program.
Finally, we round the fractional solution of the LP to obtain a feasible cut satisfying the cost constraint.
If we take the best candidate among all possible assignments, we can obtain a cut of size at most $(1 + \varepsilon)\OPT$.
Since the optimal cut size is large, the additive errors introduced by the weak regularity lemma and by the discretization are small compared to $\OPT$.

The pseudocode of the algorithm is shown in Algorithm~\ref{alg:largeOPT}. $\alpha$ is a constant to be determined in the small-optimal-value case analysis.

\begin{algorithm}[H]
  \caption{ConstrainedMinCutLarge($G = (V, E), c, \rho, \zeta, \varepsilon$)}
  \label{alg:largeOPT}
  \begin{algorithmic}[1]
    \State $A \gets$ adjacency matrix of $G$, $\delta \gets$ the minimum degree of $G$ divided by $n$
    \State $c_0 \gets \min_{v \in V} c_v / n$
    \State $L\gets \emptyset$
    \State $\eta \gets 1/10$
    \State $\varepsilon_0 \gets \alpha\varepsilon$
    \State $\kappa \gets c_0 \tfrac{\min\set{\varepsilon_0, \delta}}{30}n$
    \For {\_ in 0..$\log{(1/\eta)}$}
    \State $(d_t, S_t, T_t)_{t=1..w}$ $\gets$ WeakRegularityLemma$(A)$ with parameter $\varepsilon = \varepsilon_0/10, \gamma = \eta$
    \State $\nu \gets \tfrac{\varepsilon_0 n}{70\sqrt{27}w}$
    \For {all $(\bar{f}_1, \ldots, \bar{f}_w, \bar{g}_1, \ldots, \bar{g}_w) \in \nu K^{2w}$}
    \If {the linear relaxation of the integer problem is infeasible}
    \State continue
    \Else
    \State $x_P \gets$ a solution to the linear relaxation of the integer problem
    \State $y_P \gets \floor{x_P}$
    \State $L^* \gets$ cut corresponding to $y_P$
    \If {$c(L^*) < \rho(1 - \zeta)$}
    \State $\hat{L} \gets L^*$
    \State Add vertices to $\hat{L}$ in an arbitrary order until $c(\hat{L}) \ge \rho(1 - \zeta)$
    \ElsIf {$c(L^*) > \rho(1 + \zeta)$}
    \State $\hat{L} \gets L^*$
    \State Remove vertices from $\hat{L}$ in an arbitrary order until $c(\hat{L}) \le \rho(1 + \zeta)$
    \Else
    \State $\hat{L} \gets L^*$
    \EndIf
    \If {$\abs{\mE(\hat{L}, \overline{\hat{L}})} < \abs{\mE(L, \overline{L})}$ or $L = \emptyset$}
    \State $L \gets \hat{L}$
    \EndIf
    \EndIf
    \EndFor
    \EndFor
    \State $R \gets V \setminus L$
    \State \Return $(L, R)$
  \end{algorithmic}
\end{algorithm}

\paragraph{Proof of Theorem~\ref{thm:largeOPT}:}
By Theorem~\ref{thm:weak_regularity}, we can obtain a cut decomposition $(d_t, S_t, T_t)_{t = 1, \ldots, w}$ of the adjacency matrix $A$ with error at most $\tfrac{\varepsilon_0}{10} n^2$, width $w = O(\varepsilon_0^{-2})$ and coefficient length at most $\sqrt{27}$ in time $2^{\tilde{O}(\varepsilon_0^{-2})}$ with probability $9/10$ (we fix error parameter $\gamma$ in Theorem~\ref{thm:weak_regularity} to be $1/10$).

First, we split $V$ into at most $2^{2w}$ parts $\cP$.
Recall that $S_t, T_t \subseteq V$.
We can construct vectors $x_v = (\bbm[v \in S_1], \bbm[v \in T_1], \ldots, \bbm[v \in S_w], \bbm[v \in T_w])$ for each vertex $v \in V$.
For each $\sigma \in \set{0, 1}^{2w}$, define $P_\sigma = \set{v \in V \mid x_v = \sigma}$. Then $\cP$ forms a partition of $V$.
In the cut decomposition $(d_t, S_t, T_t)$, the contribution of each term to the cut $\mE(L, \overline{L})$ is expressed by $d_tf_tg_t$ where $f_t = \abs{L \cap S_t}$ and $g_t = \abs{\overline{L} \cap T_t}$.
Discretizing $f_t$ and $g_t$ into $\bar{f}_t := \floor{\tfrac{f_t}{\nu}} \nu$ and $\bar{g}_t := \floor{\tfrac{g_t}{\nu}} \nu$ respectively, we can approximate $\sum_{t=1}^w d_tf_tg_t$ by $\sum_{t=1}^w d_t\bar{f}_t\bar{g}_t$ with error at most $3n\nu\sum_{t=1}^w\abs{d_t}$ where $\nu = \tfrac{\varepsilon_0 n}{70\sqrt{27}w}$ (after appropriate scaling).
Now, $\bar{f}_t$ and $\bar{g}_t$ only take discrete values in $\nu K := \set{0, \nu, 2\nu, \ldots, \ceil{\frac{n}{\nu}}\nu}$, so we can enumerate all $O(1/\varepsilon_0^3)$ values for each $\bar{f}_t$ and $\bar{g}_t$.

We further split each of these $2^{2w}$ parts into $\ceil{\frac{n}{\kappa}}$ parts based on their cost $c_v$.
We fix a parameter $\kappa = c_0 \tfrac{\min\set{\varepsilon_0, \delta} n}{30}$.
We define $V_m = \set{v \in V \mid c_v \in [\kappa m, \kappa(m + 1))}$ for $m = 0, 1, \ldots, \ceil{\frac{n}{\kappa}} - 1$.
Note that $V_m$ forms a partition of $V$ and $\set{P \cap V_m \mid P \in \cP, m = 0, 1, \ldots, \ceil{\frac{n}{\kappa}} - 1}$ is also a partition of $V$.
Let $\cP' = \set{P \cap V_m \mid P \in \cP, m = 0, 1, \ldots, \ceil{\frac{n}{\kappa}} - 1}$.
Thus, each part $P \in \cP'$ is characterized by $S_t$ and $T_t$ for $(t = 1, \ldots, w)$ and its cost $\Delta_P := \kappa m$.
The number of parts in $\cP'$ is at most $2^{2w}\ceil{\frac{n}{\kappa}} = 2^{2w}\ceil{\tfrac{30}{c_0 \min\set{\varepsilon_0, \delta}}}$, which does not depend on $n$.

For each $(\bar{f}_1, \ldots, \bar{f}_w, \bar{g}_1, \ldots, \bar{g}_w) \in \nu K^{2w}$, we consider the following integer problem:
\begin{equation*}
  \begin{aligned}
    0             & \le {x_P}                                                                           & \le \abs{P} \quad (P \in \cP')                  \\
    \bar{f}_t     & \le {\sum_{P \in \cP', P \subseteq S_t} x_P}                                        & \le \bar{f}_t + \nu \quad (t = 1, 2, \ldots, w) \\
    \bar{g}_t     & \le {\sum_{P \in \cP', P \subseteq T_t} (|P| - x_P)}                                & \le \bar{g}_t + \nu \quad (t = 1, 2, \ldots, w) \\
    \rho(1-\zeta) & \le \sum_{P \in \cP'} (\Delta_P + \kappa) x_P, \quad \sum_{P \in \cP'} \Delta_P x_P & \le \rho(1+\zeta),                              \\
    x_P           & \in \mathbb{Z} \quad (P \in \cP')
  \end{aligned}
\end{equation*}
where $x_P$ is an integer variable representing the number of vertices in part $P$ assigned to cut $\mE(L, \overline{L})$.
For any cut $\mE(L, \overline{L})$ corresponding to $\bar{f}_t = \floor{\frac{\abs{L \cap S_t}}{\nu}}\nu$ and $\bar{g}_t = \floor{\frac{\abs{\overline{L} \cap T_t}}{\nu}}\nu$ for all $t = 1, \ldots, w$, there exists a feasible solution to this integer problem.
In other words, there exist parameters $(\bar{f}_1, \ldots, \bar{f}_w, \bar{g}_1, \ldots, \bar{g}_w)$ such that the optimal cut $\mE(\LOPT, \overline{\LOPT})$ corresponds to them.

To obtain an integer solution for this problem, we consider the LP relaxation of this problem and let $x_P$ be a feasible solution.
We round $x_P$ to $y_P = \floor{x_P}$ for all $P \in \cP'$.
Once we obtain $y_P$, we can construct the corresponding cut $\mE(L^*, \overline{L^*})$.
We show that the cut $\mE(L^*, \overline{L^*})$ corresponding to the rounded solution $y_P$ does not violate the cost constraint too much, and its cut size is close to the value computed from $\bar{f}_t$ and $\bar{g}_t$.

Since $x_P$ is a feasible solution to the LP relaxation, the following inequalities hold:
\begin{align*}
  \rho(1-\zeta)                  & \le \sum_{P \in \cP'} (\Delta_P + \kappa) x_P \\
  \sum_{P \in \cP'} \Delta_P x_P & \le \rho(1+\zeta).
\end{align*}

Let $z_P := x_P - \floor{x_P}$ for each $P \in \cP'$.
We have the upper and lower bounds on the cost of the rounded solution as follows:
\begin{align*}
  \sum_{P \in \cP'} (\Delta_P + \kappa) y_P
   & = \sum_{P \in \cP'} (\Delta_P + \kappa)(x_P - z_P)                                                                  \\
   & = \sum_{P \in \cP'} \pare*{\Delta_P x_P + \kappa x_P - z_P(\Delta_P + \kappa)}                                      \\
   & \le \rho(1+\zeta) + \sum_{P \in \cP'} \kappa x_P \qquad (\because \sum_{P \in \cP'} \Delta_P x_P \le \rho(1+\zeta)) \\
   & \le \rho(1+\zeta) + \kappa n. \qquad (\because \sum_{P \in \cP'} x_P \le n)
\end{align*}
\begin{align*}
  \sum_{P \in \cP'} \Delta_P y_P
   & = \sum_{P \in \cP'} \pare*{(\Delta_P + \kappa)(x_P - z_P) - \kappa \floor{x_P}}                                                                 \\
   & = \sum_{P \in  \cP'} \pare*{(\Delta_P + \kappa) x_P - z_P(\Delta_P + \kappa) - \kappa \floor{x_P}}                                              \\
   & \ge \rho(1-\zeta) - \sum_{P \in \cP'}\pare*{(\Delta_P + \kappa) + \kappa \floor{x_P}} \qquad (\because 0 \le z_P \le 1)                         \\
   & \ge \rho(1-\zeta) - \abs{\cP'}n - \kappa \abs{\cP'} - \kappa n \qquad (\because \Delta_P \le n \text{ and } \sum_{P \in \cP'}\floor{x_P} \le n) \\
   & \ge \rho(1-\zeta) - 3\kappa n. \qquad (\because \abs{\cP'} < n, \abs{\cP'} < \kappa \text{ holds for sufficiently large } n)
\end{align*}

Since $\sum_{P \in \cP'}\Delta_P y_P \le c(L^*) \le \sum_{P \in \cP'} (\Delta_P + \kappa) y_P$ holds, by rounding $x_P$ down to $y_P$, we can ensure that $\rho(1-\zeta) - 3\kappa n \le c(L^*) \le \rho(1+\zeta) + \kappa n$.

We now verify that the cut size induced by $y_P$ is close to the value computed from $\bar{f}_t$ and $\bar{g}_t$.
Since $x_P - y_P \le 1$ for each $P \in \cP'$, $\sum_{P \in \cP', P \subseteq S_t}\pare*{x_P - y_P}$ is at most $\abs{\cP'} = 2^{2w}\ceil{\tfrac{n}{\kappa}}$.
With $\sum_{P \in \cP', P \subseteq S_t} x_P \ge \bar{f}_t$, we can ensure that $\bar{f}_t - \abs{\cP'} \le \sum_{P \in \cP', P \subseteq S_t} y_P \le \bar{f}_t + \nu$.
We have a similar bound for $\sum_{P \in \cP', P \subseteq T_t} (|P| - y_P)$.
Thus, we obtain the following bounds:
\begin{align*}
  \abs{\abs{S_t \cap L^*} - \bar{f}_t} \le \nu + 2^{2w}\ceil{\tfrac{n}{\kappa}} \le 2\nu            \\
  \abs{\abs{T_t \cap \overline{L^*}} - \bar{g}_t} \le \nu + 2^{2w}\ceil{\tfrac{n}{\kappa}} \le 2\nu \\
\end{align*}
for sufficiently large $n$.

These bounds imply:
\begin{align*}
  \abs{\sum_{t=1}^w d_t |S_t \cap L^*| |T_t \cap \overline{L^*}| - \sum_{t=1}^w d_t \bar{f}_t \bar{g}_t}
   & = \abs{\sum_{t=1}^w d_t \bigl((|S_t \cap L^*| - \bar{f}_t) |T_t \cap \overline{L^*}| + \bar{f}_t (|T_t \cap \overline{L^*}| - \bar{g}_t) \bigr) } \\
   & \le \sum_{t=1}^w \abs{d_t} (2\nu n + 2\nu n)                                                                                                      \\
   & \le 4n\nu \sum_{t=1}^w \abs{d_t}                                                                                                                  \\
   & \le 4\sqrt{27}n\nu w.
\end{align*}
Thus, the total error in the cut size $\abs{\mE(L^*, \overline{L^*})}$ from discretizing $f_t$ and $g_t$ and rounding is at most $7\sqrt{27}n\nu w \le \tfrac{\varepsilon_0 n^2}{10}$.

Finally, we can adjust the rounded solution $y_P$ to satisfy the cost constraint by moving vertices.
Since $\kappa n \le c_0 \tfrac{\min\set{\varepsilon_0, \delta}}{30}n^2$ and $\rho(1 + \zeta) - \rho(1-\zeta) = 2\rho\zeta = \Omega(n^2)$, we can adjust the cut to satisfy the cost constraint as follows:
We move vertices one by one in an arbitrary order until the cost constraint is satisfied, and let $\hat{L}$ be the final cut after moving vertices.
Since the cost of each vertex is at most $n$, and the length of the constraint interval is $\Omega(n^2)$, we can always find a feasible solution by this procedure.
For example, if $c(L^*) \le \rho(1-\zeta)$ holds, then for the final moved vertex $v$ and moved set $S$, we have:
\begin{equation*}
  c(L^*) + c(S) < \rho(1-\zeta) \le c(L^*) + c(S) + c_v
\end{equation*}

Since $\rho(1-\zeta) - 3\kappa n \le c(L^*)$ holds, we have $c(S) + c_v < 3 \kappa n + n$ and $\hat{L} = L^* \cup S \cup \set{v}$ is a feasible solution satisfying the cost constraint.
The increase in the cut size by moving one vertex is bounded by its degree, which is at most $\tfrac{1}{c_0} c_v$.
More generally, moving $T \subseteq L^*$ to $\overline{L^*}$ or moving $T \subseteq \overline{L^*}$ to $L^*$ changes the cut size by at most $\tfrac{1}{c_0}c(T)$.
Hence, the total increase in the cut size by the moving procedure is at most $\tfrac{1}{c_0}(3\kappa n + n)$.
Thus, the total error for the cut size is at most $\tfrac{1}{c_0}(3\kappa n + n) + \tfrac{\varepsilon_0 n^2}{10} \le \tfrac{\varepsilon_0 n^2}{10} + \tfrac{1}{c_0}n + \tfrac{\varepsilon_0 n^2}{10} \le \tfrac{1}{c_0}n + \frac{\varepsilon_0 n^2}{5}$.
The cut size of the final solution is at most $\OPT + \tfrac{1}{c_0}n + \frac{\varepsilon_0 n^2}{5} + \tfrac{2\varepsilon_0}{10}n^2 < \OPT + \varepsilon_0 n^2$ for sufficiently large $n$ ($\tfrac{2\varepsilon_0}{10}n^2$ is the error from the weak regularity lemma).

If we set $\varepsilon_0 = \alpha \varepsilon$ and run the above algorithm, the best cut among all $\bar{f}_t$ and $\bar{g}_t$ has size at most $\OPT + \varepsilon_0 n^2 \le (1 + \varepsilon)\OPT$ with probability $9/10$.
Repeating this process $O(\log{(1/\eta)})$ times and taking the best cut among them, the success probability increases to $1 - \eta$, which completes the proof of Theorem~\ref{thm:largeOPT}.\qed

\subsection{Small OPT Case for ConstrainedMinCut Problem}
In this section, we consider the case where the optimal cut size for the \textsc{ConstrainedMinCut} problem is $\OPT \le \alpha n^2$. The value of $\alpha$ is given explicitly in the analysis below.
We denote the upper and lower cost bounds for the given \textsc{ConstrainedMinCut} instance by $C_M = \rho(1+\zeta)$ and $C_m = \rho(1-\zeta)$, respectively. The main result in this section is as follows:

\begin{thm}
  \label{thm:smallOPT} For a dense instance of the \textsc{ConstrainedMinCut} problem with optimal cut size $\OPT \le \alpha n^{2}$, with high probability, we can find a cut $\mE(L, \overline{L})$ such that $c(L) \in [\rho(1-\zeta), \rho(1+\zeta)]$ and its cut size is at most $(1 + \varepsilon)\OPT$ in time $n^{O(1)}$.
\end{thm}
Recall that $\rho' := \rho / c(V)$, which is $\Theta(1)$.

Algorithm~\ref{alg:small_case} presents the overall procedure for the small $\OPT$ case, and Algorithm~\ref{alg:adjust} presents the rebalancing procedure used in Algorithm~\ref{alg:small_case}.
The algorithm consists of the following two steps:
\begin{enumerate}
  \item Sampling and Classification.
  \item Rebalancing to satisfy the cost constraint.
\end{enumerate}
Despite being everywhere-$\delta$-dense, the graph admits an optimal cut of small size.
Hence, in the optimal cut $(L, R)$, it is possible that only a few edges cross between $L$ and $R$, and most of the edges are internal to $L$ or to $R$.
In Algorithm~\ref{alg:small_case}, we first sample a set $S$ with $\abs{S} = \Theta(\log{n})$ uniformly at random.
Since the graph is everywhere-$\delta$-dense, most of the vertices $v$ have sufficiently many neighbors in $S$.
Thus, for most vertices $v$, we can classify $v$ correctly into $L$ or $R$ by comparing the fraction of its neighbors in $S$ that lie in either $L$ or $R$.
However, due to the cost constraint, some vertices may be misclassified, and hence the resulting partition may violate the cost constraint.
Since the optimal cut size is small, the number of such vertices must be small.
Therefore, we can adjust the cut to satisfy the cost constraint without increasing the cut size by much.

\begin{algorithm}[h]
  \caption{ConstrainedMinCutSmall$(G = (V, E), c, \rho, \zeta, \varepsilon)$}
  \label{alg:small_case}
  \begin{algorithmic}[1]
    \State $\delta \gets$ the minimum degree of $G$ divided by $n$
    \State $c_0 \gets \min_{v \in V} c_v / n$
    \State $\rho' \gets \rho / c(V)$
    \State $a \gets 1/10$
    \State $\alpha \gets \min\set{\tfrac{\varepsilon}{22}(1/2 - 3a)^2 \delta^2, 6\rho' \zeta c_0 \delta /35}$
    \State $k \gets \max\set{\tfrac{2}{a^2\delta}, \tfrac{4}{\delta^2}}$
    \State Sample a sequence $S$ of $k\log{n}$ vertices uniformly at random from $V$ with replacement
    \State $L_0, R_0 \gets \emptyset, \emptyset$
    \ForAll {consistent partitions of the occurrences of $S$ into $S_L$ and $S_R$}
    \State $L, R, X \gets \emptyset, \emptyset, \emptyset$
    \ForAll {$v \in V$}
    \State Compute $\hat{p}_L(v) \gets \frac{\abs{N(v) \cap S_L}}{\abs{N(v) \cap S}}$ and $\hat{p}_R(v) \gets \frac{\abs{N(v) \cap S_R}}{\abs{N(v) \cap S}}$
    \If {$\hat{p}_L(v) > \frac 1 2 + 2a$}
    \State $L \gets L \cup \set{v}$
    \ElsIf {$\hat{p}_R(v) > \frac 1 2 + 2a$}
    \State $R \gets R \cup \set{v}$
    \Else
    \State $X \gets X \cup \set{v}$
    \EndIf
    \EndFor
    \State $(\hat{L}, \hat{R}) \gets$ Adjust$(L, R, X, \rho(1-\zeta), \rho(1+\zeta))$
    \If{$\abs{\mE(\hat{L}, \overline{\hat{L}})} < \abs{\mE(L_0, \overline{L_0})}$ or $L_0 = \emptyset$}
    \State $L_0, R_0 \gets \hat{L}, \hat{R}$
    \EndIf
    \EndFor
    \State \Return $(L_0, R_0)$
  \end{algorithmic}
\end{algorithm}

\begin{algorithm}[htbp]
  \caption{Adjust$(L, R, X, C_m, C_M)$}
  \label{alg:adjust}
  \begin{algorithmic}[1]
    \If{$c(L) + c(X) \ge C_M$}
    \State $V_u \gets |N(u) \cap L| - |N(u) \cap R|$ for all $u \in L$
    \State $V'_u \gets |N(u) \cap R| - |N(u) \cap L|$ for all $u \in X$
    \State Let $S, T$ be an optimal solution to the following optimization problem:
    \begin{equation*}
      \begin{gathered}
        \makebox[0pt][l]{\text{minimize}}\makebox[\linewidth][c]{\ensuremath{\displaystyle \sum_{u \in S} V_u + \sum_{u \in T} V'_u}}\\
        \makebox[0pt][l]{\text{s.t.}}\makebox[\linewidth][c]{\ensuremath{S \subseteq L, T \subseteq X}}\\
        \makebox[\linewidth][c]{\ensuremath{\sum_{u \in S} d_u \le \frac{\alpha n^2}{1/2 + a}}}\\
        \makebox[\linewidth][c]{\ensuremath{c(L) - c(S) + c(T) \le C_M}}
      \end{gathered}
    \end{equation*}
    \State \Return $(L \setminus S) \cup T, (R \cup S) \cup (X \setminus T)$
    \ElsIf {$c(L) < C_m$}
    \State $V_u \gets |N(u) \cap R| - |N(u) \cap L|$ for all $u \in R$
    \State $V'_u \gets |N(u) \cap L| - |N(u) \cap R|$ for all $u \in X$
    \State Let $S, T$ be an optimal solution to the following optimization problem:
    \begin{equation*}
      \begin{gathered}
        \makebox[0pt][l]{\text{minimize}}\makebox[\linewidth][c]{\ensuremath{\displaystyle \sum_{u \in S} V_u + \sum_{u \in T} V'_u}}\\
        \makebox[0pt][l]{\text{s.t.}}\makebox[\linewidth][c]{\ensuremath{S \subseteq R, T \subseteq X}}\\
        \makebox[\linewidth][c]{\ensuremath{\sum_{u \in S} d_u \le \frac{\alpha n^2}{1/2 + a}}}\\
        \makebox[\linewidth][c]{\ensuremath{C_m \le c(L) + c(S) + c(X \setminus T)}}
      \end{gathered}
    \end{equation*}

    \State \Return $(L \cup S) \cup (X \setminus T), (R \setminus S) \cup T$
    \Else
    \State $V_u \gets |N(u) \cap R| - |N(u) \cap L|$ for all $u \in X$
    \State Let $T$ be an optimal solution to the following optimization problem:
    \begin{equation*}
      \begin{gathered}
        \makebox[0pt][l]{\text{minimize}}\makebox[\linewidth][c]{\ensuremath{\displaystyle \sum_{u \in T} V_u}}\\
        \makebox[0pt][l]{\text{s.t.}}
        \makebox[\linewidth][c]{\ensuremath{T \subseteq X}}\\
      \end{gathered}
    \end{equation*}
    \State \Return $L \cup T, R \cup (X \setminus T)$
    \EndIf
  \end{algorithmic}
\end{algorithm}
\paragraph{Proof of Theorem~\ref{thm:smallOPT}:}
First of all, let $S=(s_1,\ldots,s_{k\log n})$ be the sampled sequence.
We split the occurrences of $S$ into two subsequences $S_L$ and $S_R$.
A partition $(S_L,S_R)$ of the occurrences of $S$ is called consistent if, for all $i,j$ with $s_i=s_j$, the two occurrences $s_i$ and $s_j$ are assigned to the same side.

Throughout this subsection, for a vertex set $U\subseteq V$ and a sampled sequence or subsequence $T$, the quantity $\abs{U\cap T}$ denotes the number of occurrences in $T$ whose vertex belongs to $U$.
In particular, if the same vertex appears multiple times in $T$, it is counted with its multiplicity.

Let $\hat{p}_L(v) = \frac{| N(v) \cap S_L |}{| N(v) \cap S |}$ and $p_L(v) = \frac{| N(v) \cap \LOPT |}{d_v}$ for each vertex $v \in V$ (similarly, we define $\hat{p}_R(v) := \tfrac{| N(v) \cap S_R |}{| N(v) \cap S |}$ and $p_R(v) := \tfrac{| N(v) \cap \ROPT |}{d_v}$).
The estimators $\hat{p}_L(v)$ and $\hat{p}_R(v)$ are therefore computed using the multiplicities of sampled neighbors in $S$.
We show that, when the consistent partition $(S_L,S_R)$ agrees with the optimal cut $(\LOPT,\ROPT)$ on the sampled vertices, with high probability, $\hat{p}_L(v)$ is close to $p_L(v)$ for all $v \in V$.

We show the following key lemma.
\begin{lem}
  If we take $k \ge \max(\frac{2}{a^2\delta}, \frac{4}{\delta^2})$ for a constant $a > 0$, then $| \hat{p}_L(v) - p_L(v) | \le a$ holds with high probability (over the choice of $S$) for all $v \in V$.
\end{lem}
\begin{proof}
  First, we consider the conditional probability $\Pr[| \hat{p}_L(v) - p_L(v) | \ge a \mid | N(v) \cap S | = m]$.
  We define random variables $X_i, Y_i \ (i = 1, \ldots, m)$ as follows:
  \begin{equation*}
    Y_i \text{ is sampled uniformly at random from } N(v)
  \end{equation*}
  \begin{equation*}
    X_i =
    \begin{cases}
      1 & \text{$Y_i \in \LOPT$}     \\
      0 & \text{$Y_i \notin \LOPT$}.
    \end{cases}
  \end{equation*}
  Conditional on $\abs{N(v)\cap S}=m$, the $m$ occurrences of $S$ that lie in $N(v)$ have the same distribution as independent samples $Y_1,\ldots,Y_m$ drawn uniformly from $N(v)$ with replacement.
  Note that these occurrences form a multiset.

  Specifically, we have $\E[X_i] = \frac 1 {d_v} \sum_{u \in V} \bbm[u \in N(v) \cap \LOPT] =  \frac{|N(v) \cap \LOPT|} {d_v} = p_L(v)$.
  By Hoeffding's inequality, we have:
  \begin{equation*}
    \Pr[|\hat{p}_L(v) - p_L(v)|\ge a \mid | N(v) \cap S| = m]
    \le 2\exp(-2ma^2).
  \end{equation*}

  We denote $I = [(\frac {d_v} {n} - u)|S|, (\frac {d_v} {n} + u)|S|]$.
  Then we have:
  \begin{align*}
    \Pr[|\frac {|N(v) \cap S|}{|S|} - \frac{d_v}{n}| \ge u]
     & = \Pr[|N(v) \cap S| \notin I] \\
     & \le 2\exp(-2|S|u^2)           \\
     & = 2n^{-\frac{k\delta^2}{2}}.
  \end{align*}

  We set $u = \frac \delta 2$ and $I = [\frac{k\delta\log{n}}{2}, (1 + \frac{\delta}{2})k\log{n}]$. Then we have:
  \begin{align*}
    \max_{m \in I} \Pr[|\hat{p}_L(v) - p_L(v)| \ge a \mid | N(v) \cap S| = m]
     & \le 2\exp(-k\delta a^2\log{n}) \\
     & = 2n^{-k\delta a^2}.
  \end{align*}

  \begin{align*}
    \Pr[\abs{\hat{p}_L(v) - p_L(v)} \ge a]
     & = \sum_{m \in I} \Pr[\abs{\hat{p}_L(v) - p_L(v)} \ge a \mid | N(v) \cap S| = m] \Pr[| N(v) \cap S| = m]          \\
     & + \sum_{m \notin I} \Pr[\abs{\hat{p}_L(v) - p_L(v)} \ge a \mid | N(v) \cap S| = m] \Pr[| N(v) \cap S| = m]       \\
     & \le \max_{m \in I} \Pr[\abs{\hat{p}_L(v) - p_L(v)} \ge a \mid | N(v) \cap S| = m] + \Pr[| N(v) \cap S| \notin I] \\
     & \le 2n^{-k\delta a^2} + 2n^{-\frac{k\delta^2}{2}}                                                                \\
     & \le 4n^{-\min(k\delta a^2, \frac{k\delta^2}{2})}.
  \end{align*}

  By union bound, we have:
  \begin{align*}
    \Pr[^\exists v \in V, |\hat{p}_L(v) - p_L(v)| \ge a]
     & \le \sum_{v \in V} \Pr[|\hat{p}_L(v) - p_L(v)| \ge a] \\
     & \le 4n^{1 - \min(k\delta a^2, \frac{k\delta^2}{2})}.
  \end{align*}
  Thus, if we set $k \ge \max(\frac{2}{a^2\delta}, \frac{4}{\delta^2})$, $|\hat{p}_L(v) - p_L(v)| \le a$ holds with high probability for all $v \in V$.
\end{proof}
In what follows, we assume that the above lemma holds.

Let $L_0 = L \cap \LOPT$, $R_0 = R \cap \ROPT$, $U_L = L \cap \ROPT$, and $U_R = R \cap \LOPT$.
We obtain the following lemma concerning the volumes of $U_L$ and $U_R$, respectively.
\begin{clm}
  \label{clm:volume_ub}
  $\vol(U_L), \vol(U_R) \le \frac{\OPT}{1/2 + a}$ and $|U_L|, |U_R| \le \frac{\OPT}{(1/2 + a)\delta n}$ hold.
\end{clm}
\begin{proof}
  If $v \in U_L$, then $p_L(v) \ge \frac 1 2 + a$ holds.
  \begin{align*}
    \OPT
     & \ge \sum_{v \in U_L} |N(v) \cap \LOPT|  \\
     & \ge \sum_{v \in U_L} (\frac 1 2 + a)d_v \\
     & = (\frac 1 2 + a) \vol(U_L).
  \end{align*}
  Thus, we have $\vol(U_L) \le \frac{\OPT}{1/2 + a}$.
  $\vol(U_R)$ can be bounded in the same way.
\end{proof}
This claim implies that $L, R$ are close to $\LOPT, \ROPT$, respectively.
In addition, we can bound the volume of $X$ in a similar manner.
We denote $X_L = X \cap \LOPT$ and $X_R = X \cap \ROPT$.
\begin{clm}
  \label{clm:volume_x_ub}
  $\vol(X) \le \frac{2\OPT}{1/2 - 3a}$ and $|X| \le \frac{2\OPT}{(1/2 - 3a)\delta n}$ hold.
\end{clm}
\begin{proof}
  If $v\in X_L$, then $\hat{p}_L(v)\le \frac{1}{2}+2a$, and hence $p_L(v)\le \frac{1}{2}+3a$ by the key lemma. Therefore $p_R(v)\ge \frac{1}{2}-3a$.
  Similarly, if $v\in X_R$, then $p_L(v)\ge \frac{1}{2}-3a$.
  Thus, every vertex in $X$ contributes at least $(\frac{1}{2}-3a)d_v$ to the optimal cut, counted from its own side.

  \begin{align*}
    2\OPT
     & \ge \sum_{v \in X_L} \abs{N(v) \cap \ROPT} + \sum_{v \in X_R} \abs{N(v) \cap \LOPT} \\
     & \ge \sum_{v \in X_L} (\frac{1}{2} - 3a)d_v + \sum_{v \in X_R} (\frac{1}{2} - 3a)d_v \\
     & \ge (\frac{1}{2} - 3a) \vol(X).
  \end{align*}
  Thus, we have $\vol(X) \le \frac{2\OPT}{1/2 - 3a}$ and $|X| \le \frac{2\OPT}{(1/2 - 3a)\delta n}$.
\end{proof}
Intuitively, the vertices in $X$ make a large contribution to the optimal cut.
Since the optimal cut is bounded by $\alpha n^2$, their total volume and size cannot be too large.

We now show that, in the first case of Algorithm~\ref{alg:adjust}, where $c(L) + c(X) \ge C_M$, we can obtain a cut of size at most $(1 + \varepsilon)\OPT$.
\begin{lem}
  \label{lem:adjust}
  If $c(L) + c(X) \ge C_M$,
  we can find sets $S \subseteq L$ and $T \subseteq X$ in polynomial time such that
  \begin{equation*}
    \abs{\mE(S, L)} - \abs{\mE(S, R)} + \abs{\mE(T, R)} + \abs{\mE(X \setminus T, L)} \le \abs{\mE(U_L, L)} - \abs{\mE(U_L, R)} + \abs{\mE(X_L, R)} + \abs{\mE(X_R, L)}
  \end{equation*}
  holds and the adjusted cut $((L \setminus S) \cup T, (R \cup S) \cup (X \setminus T))$ is a feasible solution for the original \textsc{ConstrainedMinCut} instance.
\end{lem}
\begin{proof}
  Consider the following auxiliary optimization problem:
  \begin{equation*}
    \begin{gathered}
      \makebox[0pt][l]{\text{minimize}}\makebox[\linewidth][c]{\ensuremath{\displaystyle \sum_{u \in S} V_u + \sum_{u \in T} V'_u}}\\
      \makebox[0pt][l]{\text{s.t.}}\makebox[\linewidth][c]{\ensuremath{S \subseteq L, T \subseteq X}}\\
      \makebox[\linewidth][c]{\ensuremath{\sum_{u \in S} d_u \le \frac{\alpha n^2}{1/2 + a}}}\\
      \makebox[\linewidth][c]{\ensuremath{c(L) - c(S) + c(T) \le C_M}}
    \end{gathered}
  \end{equation*}
  where $V_u = \abs{N(u) \cap L} - \abs{N(u) \cap R}$ for $u \in L$ and $V'_u = \abs{N(u) \cap R} - \abs{N(u) \cap L}$ for $u \in X$.

  We claim that $(S, T) = (U_L, X_L)$ is feasible for this problem.
  By Claim~\ref{clm:volume_ub}, we have
  \begin{equation*}
    \sum_{u\in U_L} d_u = \vol(U_L)
    \le \frac{\OPT}{1/2+a}
    \le \frac{\alpha n^2}{1/2+a}.
  \end{equation*}
  Moreover,
  \begin{equation*}
    c(L)-c(U_L)+c(X_L)
    = c(L_0)+c(X_L)
    = c(\LOPT)-c(U_R)
    \le c(\LOPT)
    \le C_M.
  \end{equation*}
  Thus $(U_L,X_L)$ is feasible for the auxiliary problem.

  We emphasize that the lower cost constraint is intentionally not imposed in this auxiliary problem.
  Indeed, the natural solution $(S, T) = (U_L, X_L)$ satisfies the upper constraint, but it may violate the lower constraint, since $c(L)-c(U_L)+c(X_L) = c(L_0)+c(X_L) = c(\LOPT)-c(U_R)$.
  Thus imposing the lower bound would prevent us from comparing the optimum with the feasible pair $(U_L,X_L)$.
  Instead, we show that every feasible solution of the auxiliary problem automatically satisfies the lower bound.

  Let $S^*, T^*$ be the optimal solution for this problem.
  Since
  \begin{align*}
    \abs{\mE(T, R)} + \abs{\mE(X \setminus T,L)}
     & = \sum_{u \in T} \abs{N(u) \cap R} + \sum_{u \in X \setminus T} \abs{N(u) \cap L} \\
     & = \sum_{u \in X} \abs{N(u) \cap L} + \sum_{u \in T}V'_u
  \end{align*}
  holds, we have
  \begin{align*}
    \abs{\mE(S^*, L)} - \abs{\mE(S^*, R)} & + \abs{\mE(T^*, R)} + \abs{\mE(X \setminus T^*,L)}                                             \\
                                          & \le \abs{\mE(U_L, L)} - \abs{\mE(U_L, R)} + \abs{\mE(X_L, R)} + \abs{\mE(X \setminus X_L, L)}.
  \end{align*}
  Here, we used the fact that $X = X_L \sqcup X_R$.

  Let $(S, T)$ be any feasible solution of the auxiliary problem.
  Since
  \begin{equation*}
    \sum_{u \in S} d_u \le \frac{\alpha n^2}{1/2 + a}
  \end{equation*}
  and every vertex has degree at least $\delta n$, we have
  \begin{equation*}
    \abs{S} \le \frac{\alpha n}{\delta(1/2 + a)},
    \qquad
    c(S) \le \frac{\alpha n^2}{\delta(1/2 + a)}.
  \end{equation*}
  Moreover, by Claim~\ref{clm:volume_x_ub},
  \begin{equation*}
    \abs{X} \le \frac{2\OPT}{(1/2 - 3a)\delta n}
    \le \frac{2\alpha n}{(1/2 - 3a)\delta},
  \end{equation*}
  and hence
  \begin{equation*}
    c(X) \le \frac{2\alpha n^2}{(1/2 - 3a)\delta}.
  \end{equation*}
  Using the assumption $c(L) + c(X) \ge C_M$, we obtain
  \begin{align*}
    c(L) - c(S) + c(T)
     & \ge c(L) - c(S)                                                                      \\
     & \ge c(L) + c(X) - c(S) - c(X)                                                        \\
     & \ge C_M - \frac{\alpha n^2}{\delta(1/2 + a)} - \frac{2\alpha n^2}{\delta(1/2 - 3a)}.
  \end{align*}
  By the choice $\alpha \le \tfrac{6\rho' \zeta c_0 \delta}{35}$, we have
  \begin{equation*}
    \frac{\alpha n^2}{\delta(1/2 + a)} + \frac{2\alpha n^2}{\delta(1/2 - 3a)}
    \le 2\rho' \zeta c_0 n^2
    \le C_M - C_m,
  \end{equation*}
  where the last inequality follows from $c(V) \ge c_0 n^2$.
  Therefore, $c(L) - c(S) + c(T) \ge C_m$.
  The upper bound $c(L) - c(S) + c(T) \le C_M$ is imposed explicitly in the auxiliary problem.
  Hence every feasible solution of the auxiliary problem yields a feasible cut $((L \setminus S) \cup T, (R \cup S) \cup (X \setminus T))$ for the original \textsc{ConstrainedMinCut} instance.

  Finally, we show that the above problem can be solved in polynomial time by dynamic programming.
  First, we solve the following two knapsack problems separately:
  \begin{equation*}
    \begin{gathered}
      \makebox[0pt][l]{\text{minimize}}\makebox[\linewidth][c]{\ensuremath{\displaystyle \sum_{u \in S} V_u}}\\
      \makebox[0pt][l]{\text{s.t.}}\makebox[\linewidth][c]{\ensuremath{S \subseteq L}}\\
      \makebox[\linewidth][c]{\ensuremath{c(S) = C}}\\
      \makebox[\linewidth][c]{\ensuremath{\sum_{u \in S} d_u \le \tfrac{\alpha n^2}{1/2 + a}}}
    \end{gathered}
  \end{equation*}
  \begin{equation*}
    \begin{gathered}
      \makebox[0pt][l]{\text{minimize}}\makebox[\linewidth][c]{\ensuremath{\displaystyle \sum_{u \in T} V'_u}}\\
      \makebox[0pt][l]{\text{s.t.}}\makebox[\linewidth][c]{\ensuremath{T \subseteq X}}\\
      \makebox[\linewidth][c]{\ensuremath{c(T) = C'}}
    \end{gathered}
  \end{equation*}
  Now, we compute these optimal solutions $\min\set{\sum_{u \in S} V_u \mid S \subseteq L, c(S) = C, \sum_{u \in S} d_u \le \tfrac{\alpha n^2}{1/2 + a}}$ and $\min\set{\sum_{u \in T} V'_u \mid T \subseteq X, c(T) = C'}$ for all $C = 0, \ldots, c(L)$ and $C' = 0, \ldots, c(X)$.
  The first minimization is solved by dynamic programming over both the cost and the degree.
  Next, we combine the two solutions to obtain the optimal solution to the auxiliary problem.
  \begin{equation*}
    \min_{\substack{C = 0, \ldots, c(L)\\ C' = 0, \ldots, c(X) \\  c(L) - C + C' \le C_M}} \left\{\min\set{\sum_{u \in S} V_u \mid S \subseteq L, c(S) = C, \sum_{u \in S} d_u \le \tfrac{\alpha n^2}{1/2 + a}} + \min\set{\sum_{u \in T} V'_u \mid T \subseteq X, c(T) = C'}\right\}
  \end{equation*}
  Since $c(L), c(X), \vol(L) \le n^2$, the dynamic programming procedures and the final minimization run in polynomial time.

  The dynamic program returns a pair attaining the above minimum, which is precisely an optimal solution to the auxiliary problem.
  Therefore, by comparing it with the feasible solution $(U_L,X_L)$, we obtain the desired inequality.
  Moreover, by the feasibility argument above, the adjusted cut $((L\setminus S^*)\cup T^*,(R\cup S^*)\cup (X\setminus T^*))$ satisfies the original cost constraint.
  This completes the proof.
\end{proof}

\begin{figure}[h]
  \begin{minipage}[b]{0.32\linewidth}
    \centering
    \includegraphics[keepaspectratio, scale=0.22]{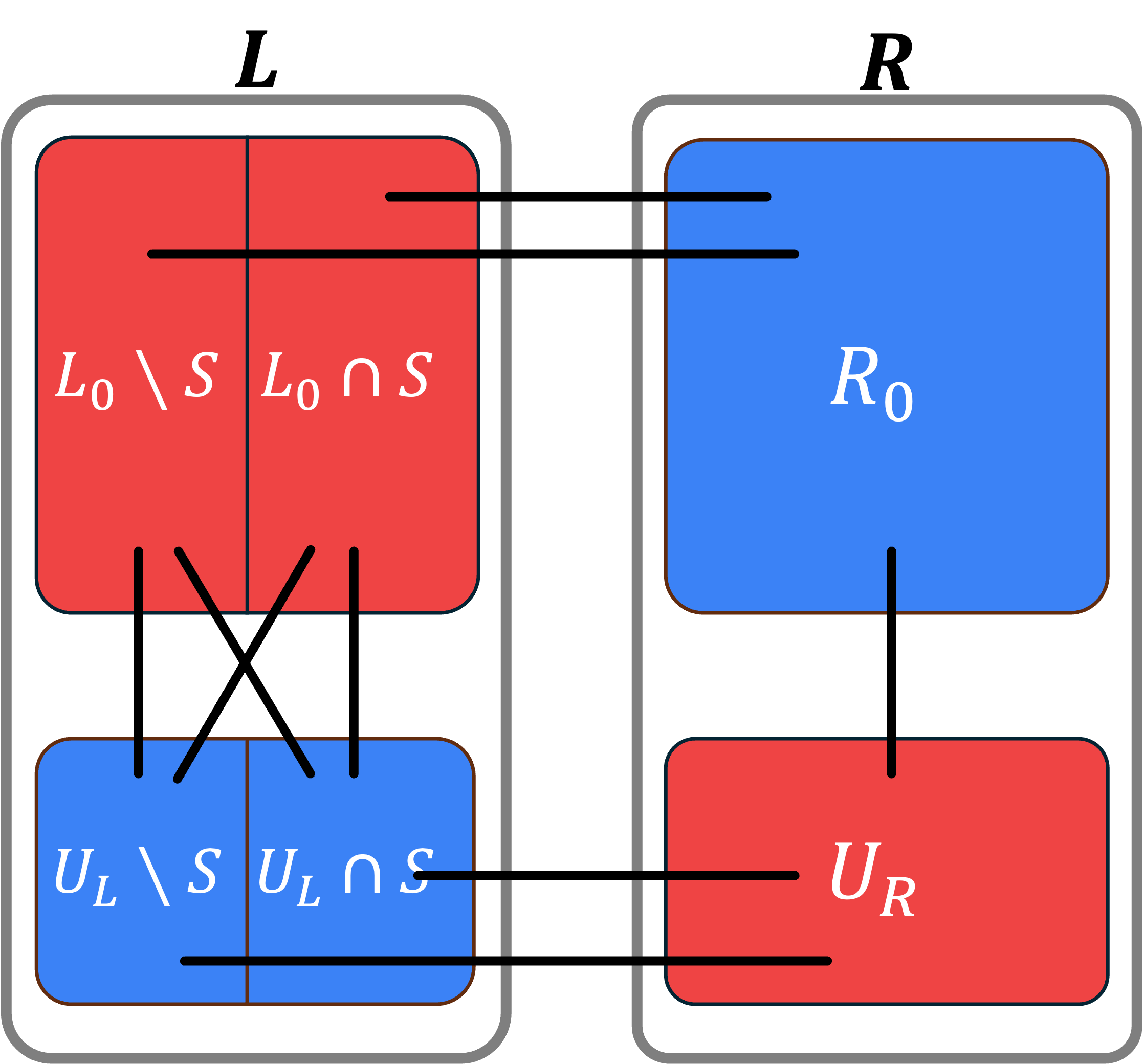}
    \subcaption{Optimal cut}
    \label{fig-opt}
  \end{minipage}
  \begin{minipage}[b]{0.32\linewidth}
    \centering
    \includegraphics[keepaspectratio, scale=0.22]{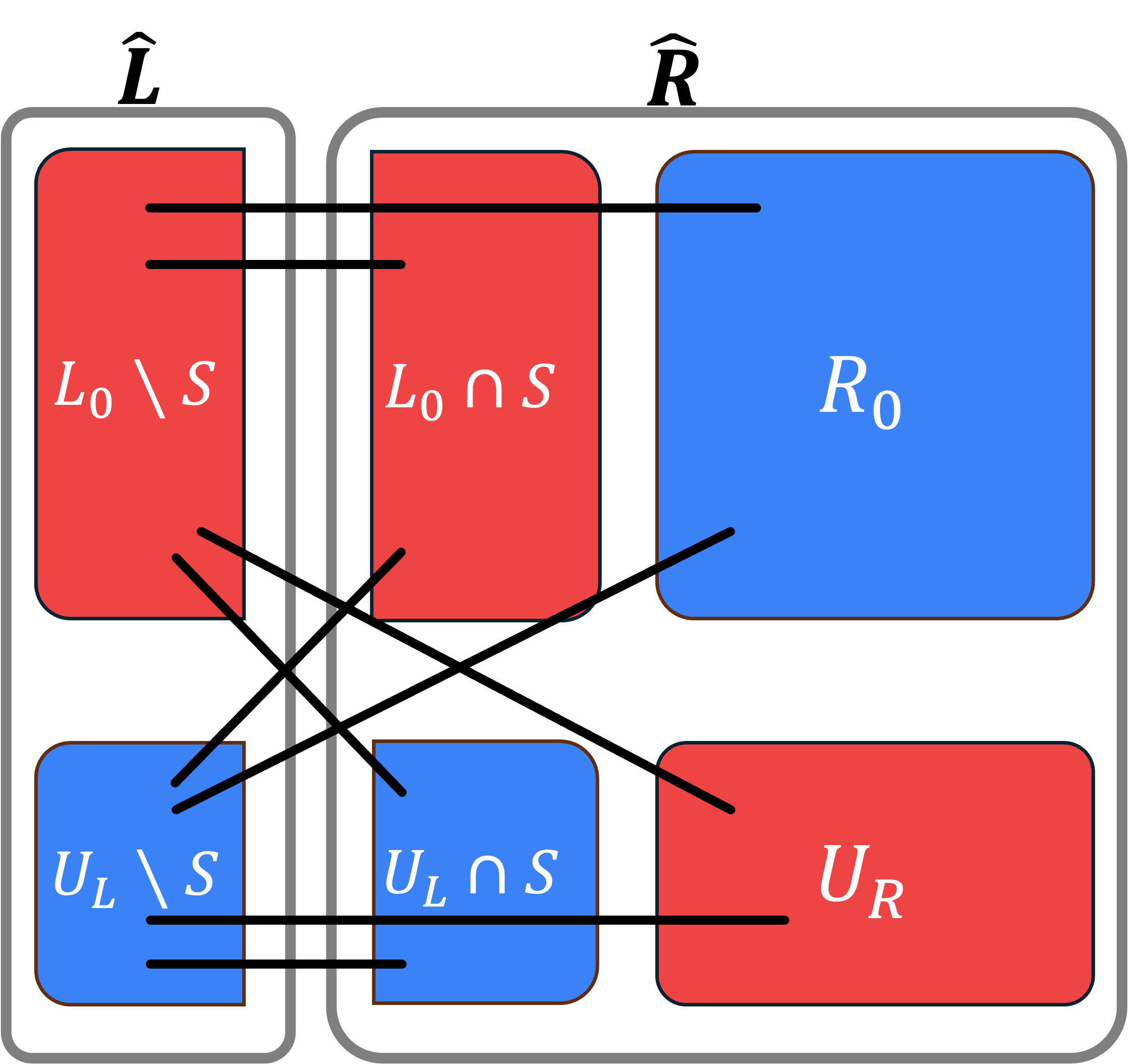}
    \subcaption{Adjusted cut}
    \label{fig-alg}
  \end{minipage}
  \begin{minipage}[b]{0.32\linewidth}
    \centering
    \includegraphics[keepaspectratio, scale=0.22]{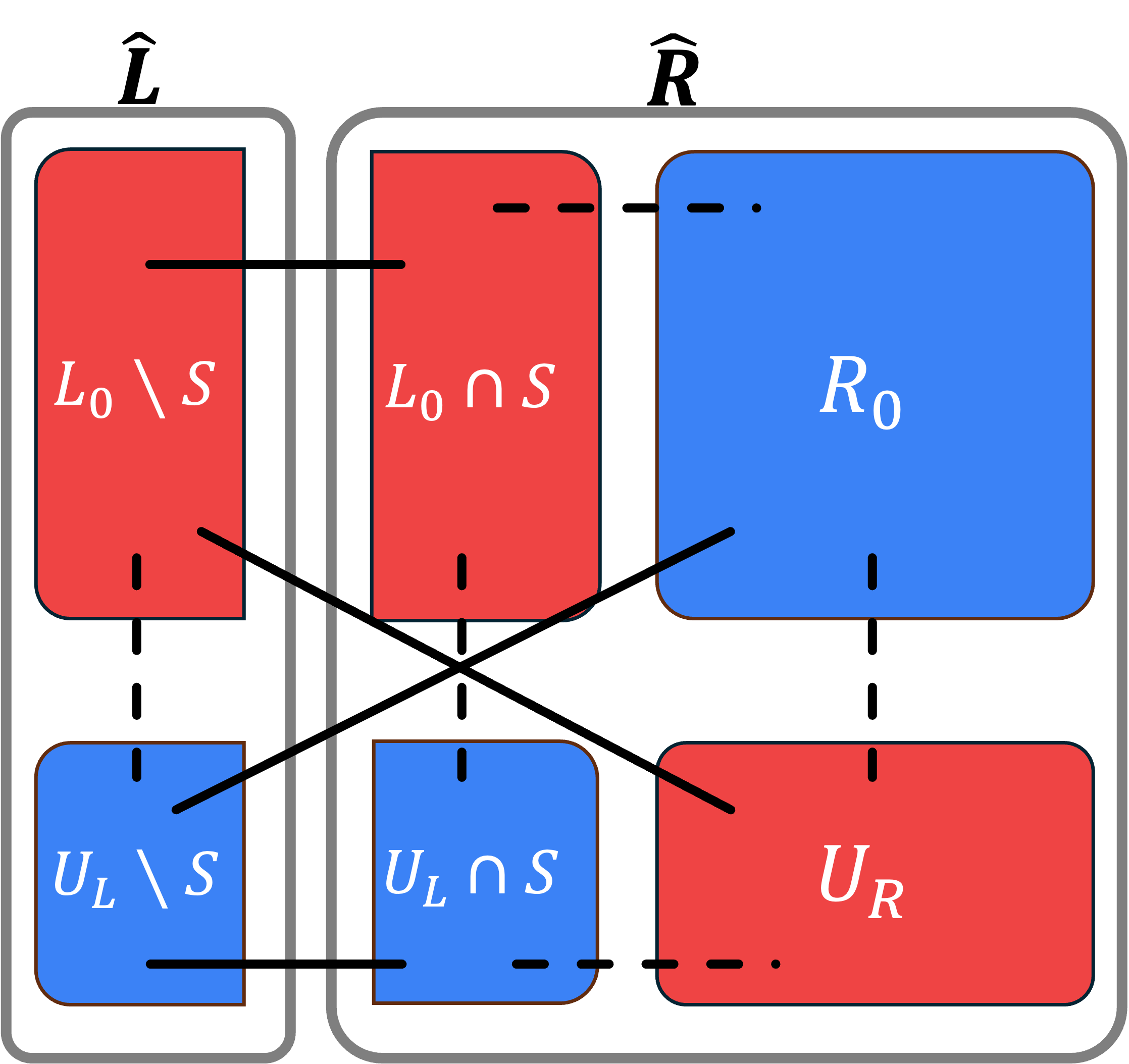}
    \subcaption{Difference}
    \label{fig-diff}
  \end{minipage}
  \caption{Evaluation of cut edges}
  \label{fig-eval-edges}
\end{figure}

Lemma~\ref{lem:adjust} shows that when $c(L) + c(X) \ge C_M$, we can find a set $S$ and move it from $L$ to $R$.
Figure~\ref{fig-eval-edges} shows the cut edges of the optimal cut, the adjusted cut by moving $S$ from $L$ to $R$, and the difference between them.
The red areas are subsets of $\LOPT$ and the blue areas are subsets of $\ROPT$.
In Figure~\ref{fig-diff}, the solid line indicates the actual cut edges and the dashed line indicates the optimal cut edges.
Let $\OPT[L \cup R]$ denote the number of optimal cut edges with both endpoints in $L \cup R$, i.e., $\OPT[L \cup R] := \abs{\mE(\LOPT \cap (L \cup R), \ROPT \cap (L \cup R))}$.
Thus, after adjusting $L$ and $R$ by Algorithm~\ref{alg:adjust}, the cut size of $(L \setminus S, R \cup S)$ is evaluated as follows:
\begin{align*}
  \abs{\mE(L \setminus S, R \cup S)}
   & = \OPT[L \cup R]                                                                                                                                \\
   & \quad + \abs{\mE(S, L\setminus S)} - (\abs{\mE(L_0 \cap S, R_0)} + \abs{\mE(U_L \cap S, U_R)}) - \abs{\mE(U_L, L_0)}                            \\
   & \quad + \abs{\mE(L_0 \setminus S, U_R)} - \abs{\mE(R_0, U_R)}                                                                                   \\
   & \quad + \abs{\mE(U_L \setminus S, R_0)}                                                                                                         \\
   & = \OPT[L \cup R]                                                                                                                                \\
   & \quad + \abs{\mE(S, L)} - (\abs{\mE(S, R)} - \abs{\mE(L_0 \cap S, U_R)} - \abs{\mE(U_L \cap S, R_0)}) - (\abs{\mE(U_L, L)} - \abs{\mE(U_L, R)}) \\
   & \quad -\abs{\mE(S, S)} + \abs{\mE(U_L, U_L)}                                                                                                    \\
   & \quad + \abs{\mE(L_0 \setminus S, U_R)} - \abs{\mE(R_0, U_R)}                                                                                   \\
   & \quad + \abs{\mE(U_L \setminus S, R_0)} - \abs{\mE(U_L, R)}                                                                                     \\
   & = \OPT[L \cup R]                                                                                                                                \\
   & \quad + \abs{\mE(S, L)} - (\abs{\mE(S, R)}) - (\abs{\mE(U_L, L)} - \abs{\mE(U_L, R)})                                                           \\
   & \quad -\abs{\mE(S, S)} + \abs{\mE(U_L, U_L)}                                                                                                    \\
   & \quad + \abs{\mE(L_0, U_R)} - \abs{\mE(R_0, U_R)}                                                                                               \\
   & \quad + \abs{\mE(U_L, R_0)} - \abs{\mE(U_L, R)}.
\end{align*}
We use the facts that $\abs{\mE(U_L, L_0)} = \abs{\mE(U_L, L)} - \abs{\mE(U_L, U_L)}$, $\abs{\mE(L_0 \setminus S, U_R)} + \abs{\mE(L_0 \cap S, U_R)} = \abs{\mE(L_0, U_R)}$, and $\abs{\mE(U_L \setminus S, R_0)} + \abs{\mE(U_L \cap S, R_0)} = \abs{\mE(U_L, R_0)}$.
Since $R_0 \subseteq R$, $\abs{\mE(U_L, R_0)} - \abs{\mE(U_L, R)} \le 0$ holds.

For $v \in U_R$, $p_L(v) \le (\frac{1}{2} - a)$ holds.
\begin{align*}
  \abs{N(v) \cap \LOPT}
   & = \abs{\mE(L_0, \set{v})} + \abs{\mE(\LOPT \setminus L_0, \set{v})}                         \\
   & \le (\frac {1} {2} - a) d_v \qquad (\because \text{definition of } p_L(v))                  \\
   & = (\frac{1}{2} - a)(\abs{\mE(L_0, \set{v})} + \abs{\mE(R_0, \set{v})}                       \\
   & \quad + \abs{\mE(\LOPT \setminus L_0, \set{v})} + \abs{\mE(\ROPT \setminus R_0, \set{v})}).
\end{align*}

By rearranging the above inequality, we have:
\begin{align*}
  \abs{\mE(L_0, \{v\})} - \abs{\mE(R_0, \{v\})}
   & \le -\abs{\mE(\LOPT \setminus L_0, \{v\})} + \abs{\mE(\ROPT \setminus R_0, \{v\})} - 2ad_v                         \\
   & \le |\ROPT \setminus R_0|. \qquad (\because \abs{\mE(\ROPT \setminus R_0, \set{v})} \le \abs{\ROPT \setminus R_0})
\end{align*}
\begin{align*}
  \abs{\mE(L_0, U_R)} - \abs{\mE(R_0, U_R)}
   & = \sum_{v \in U_R} (\abs{\mE(L_0, \{v\})} - \abs{\mE(R_0, \{v\})})                                               \\
   & \le \abs{\ROPT \setminus R_0} \abs{U_R} \qquad (\because \text{we use the above inequality for each } v \in U_R) \\
   & \le \abs{\ROPT \setminus R_0} \abs{\LOPT \setminus L_0} \qquad (\because U_R \subseteq \LOPT \setminus L_0)      \\
   & \le \frac {\OPT^2}{(1/2 - 3a)^2 \delta^2 n^2}.
\end{align*}
For $v \in \ROPT \setminus R_0$, $p_R(v) \le \tfrac{1}{2} + 3a$ (i.e., $p_L(v) \ge \tfrac 1 2 - 3a$) holds.
Thus, the last inequality is derived in the same way as in Claims~\ref{clm:volume_ub} and~\ref{clm:volume_x_ub}.
Thus, we have
\begin{align*}
  \abs{\mE(L \setminus S, R \cup S)}
   & \le \OPT[L \cup R] + |U_L|^2 + \frac {\OPT^2}{(1/2 - 3a)^2 \delta^2 n^2}       \\
   & + \abs{\mE(S, L)} - \abs{\mE(S, R)} - (\abs{\mE(U_L, L)} - \abs{\mE(U_L, R)})  \\
   & \le \OPT[L \cup R] + \frac {2\OPT^2}{(1/2 - 3a)^2 \delta^2 n^2}                \\
   & + \abs{\mE(S, L)} - \abs{\mE(S, R)} - (\abs{\mE(U_L, L)} - \abs{\mE(U_L, R)}).
\end{align*}

Let $\OPT_X$ be the number of cut edges in the optimal cut with at least one endpoint in $X$.
Next, we evaluate the contributions of $T, X\setminus T$ to the cut size.

We define $\mu := \tfrac{4\OPT^2}{(1/2 - 3a)^2\delta^2 n^2}$.
Note that for any two sets $A, B$ such that $\abs{A}, \abs{B} \le \tfrac{2\OPT}{(1/2 - 3a)\delta n}$, $\abs{\mE(A, B)} \le \abs{A}\abs{B} \le \mu$ holds.

Recall that $T \subseteq X$, $\abs{X} \le \tfrac{2\OPT}{(1/2 - 3a) \delta n}$, $\abs{S} \le \tfrac{\alpha n^2}{(1/2 + a)\delta n}$, and $\abs{U_L}, \abs{U_R} \le \tfrac{\OPT}{(1/2 + a)\delta n}$ hold with high probability.
Thus, we have the following evaluation:
\begin{align*}
   & \abs{\mE(T, R \cup S)} + \abs{\mE(X \setminus T, L \setminus S)} + \abs{\mE(T, X \setminus T)}                                                                                         \\
   & = \abs{\mE(T, R)} + \abs{\mE(T, S)} + \abs{\mE(X \setminus T, L)} - \abs{\mE(X \setminus T, S)} + \abs{\mE(T, X \setminus T)}                                                          \\
   & \le \abs{\mE(T, R)} + \abs{\mE(X \setminus T, L)} + \mu + \tfrac{2\alpha \OPT}{(1/2 -3a)^2\delta^2} - \abs{\mE(X \setminus T, S)}                                                      \\
   & \le (\abs{\mE(T, R)} + \abs{\mE(X \setminus T, L)} - \abs{\mE(X_L, R)} - \abs{\mE(X_R, L)})+ (\abs{\mE(X_L, R)} + \abs{\mE(X_R, L)}) + \mu + \tfrac{2\alpha \OPT}{(1/2 -3a)^2\delta^2} \\
   & \le (\abs{\mE(X_L, R_0)} + \abs{\mE(X_L, U_R)} + \abs{\mE(X_R, L_0)} + \abs{\mE(X_R, U_L)}) + \mu + \tfrac{2\alpha \OPT}{(1/2 -3a)^2\delta^2}                                          \\
   & + (\abs{\mE(T, R)} + \abs{\mE(X \setminus T, L)} - \abs{\mE(X_L, R)} - \abs{\mE(X_R, L)})                                                                                              \\
   & \le \OPT_X + \abs{\mE(X_L, U_R)} + \abs{\mE(X_R, U_L)} - (\abs{\mE(X_L, U_L)} + \abs{\mE(X_R, U_R)}) + \mu + \tfrac{2\alpha \OPT}{(1/2 -3a)^2\delta^2}                                 \\
   & + (\abs{\mE(T, R)} + \abs{\mE(X \setminus T, L)} - \abs{\mE(X_L, R)} - \abs{\mE(X_R, L)})                                                                                              \\
   & \le \OPT_X + 3\mu + \tfrac{2\alpha \OPT}{(1/2 -3a)^2\delta^2} + (\abs{\mE(T, R)} + \abs{\mE(X \setminus T, L)} - \abs{\mE(X_L, R)} - \abs{\mE(X_R, L)}).
\end{align*}

Therefore, the total cut size is evaluated as follows:
\begin{align*}
  \mathrm{CUT}
   & = \abs{\mE(L \setminus S, R \cup S)} + \abs{\mE(T, R \cup S)} + \abs{\mE(X \setminus T, L \setminus S)} + \abs{\mE(T, X \setminus T)}                 \\
   & \le \OPT[L \cup R] + \frac {2\OPT^2}{(1/2 - 3a)^2 \delta^2 n^2} + \abs{\mE(S, L)} - \abs{\mE(S, R)} - (\abs{\mE(U_L, L)} - \abs{\mE(U_L, R)})         \\
   & + \OPT_X + 3\mu + \tfrac{2\alpha \OPT}{(1/2 -3a)^2\delta^2} + (\abs{\mE(T, R)} + \abs{\mE(X \setminus T, L)} - \abs{\mE(X_L, R)} - \abs{\mE(X_R, L)}) \\
   & \le \OPT + 5\mu + \tfrac{2\alpha \OPT}{(1/2 -3a)^2\delta^2} \qquad (\because \text{Lemma}~\ref{lem:adjust})                                           \\
   & \le \OPT\pare*{1 + \frac{20\OPT + 2\alpha n^2}{(1/2 - 3a)^2 \delta^2 n^2}}                                                                            \\
   & < \OPT\pare*{1 + 22\frac{\alpha}{(1/2 - 3a)^2 \delta^2}}                                                                                              \\
   & \le \OPT(1 + \varepsilon).
\end{align*}
where $\alpha = \min\set{\frac{\varepsilon}{22}(1/2 - 3a)^2 \delta^2, 6 \rho' \zeta c_0\delta / 35}$. Moreover, under the balanced parameter assumption $\rho' > \tfrac{c_0 \varepsilon \delta^2}{2(1+\zeta)}$, we have $\alpha > \min\set{\frac{\varepsilon}{22}(1/2 - 3a)^2 \delta^2, \tfrac{3 c_0^2}{35}\varepsilon \tfrac{\zeta}{1 + \zeta}\delta^3} = \Theta(\varepsilon)$.
This completes the proof of the first case of Algorithm~\ref{alg:adjust}.
For the second case, where $c(L) < C_m$, we can proceed similarly.
For the third case of Algorithm~\ref{alg:adjust}, we can prove a lemma similar to Lemma~\ref{lem:adjust}.
\begin{lem}
  \label{lem:adjust3}
  If $C_m \le c(L) \le C_M - c(X)$,
  we can find a set $T \subseteq X$ in polynomial time such that
  \begin{equation*}
    \abs{\mE(T, R)} + \abs{\mE(X \setminus T,L)} \le \abs{\mE(X_L, R)} + \abs{\mE(X_R, L)}
  \end{equation*}
  holds and the adjusted cut $(L \cup T, R \cup (X \setminus T))$ is a feasible solution for the original \textsc{ConstrainedMinCut} instance.
\end{lem}
\begin{proof}
  For any $T \subseteq X$, the adjusted cut $(L \cup T, R \cup (X \setminus T))$ satisfies the cost constraint because the cost of $L$ is already between $C_m$ and $C_M - c(X)$.
  Since $(X_L, X_R)$ is a feasible solution for the optimization problem in Algorithm~\ref{alg:adjust}, the optimal solution $T^*$ satisfies
  \begin{equation*}
    \abs{\mE(T^*, R)} + \abs{\mE(X \setminus T^*,L)} \le \abs{\mE(X_L, R)} + \abs{\mE(X_R, L)}.
  \end{equation*}
  This can be solved in polynomial time by picking exactly those vertices $u \in X$ with $V_u = \abs{N(u) \cap R} - \abs{N(u) \cap L} < 0$.
\end{proof}
Using this lemma, we can evaluate the cut size as in the first case.
This completes the proof of Theorem~\ref{thm:smallOPT}.\qed

\medskip

Note that the assumption that the input graph is everywhere-$\delta$-dense is used in two places.
First, we use it to ensure that the sampling procedure works well.
If some vertices have very small degrees, the estimators $\hat{p}_L, \hat{p}_R$ may not be accurate.
Second, we use it to bound the sizes of $U_L, U_R$, and $X$.
Under this assumption, we can bound the number of vertices whose contribution to the optimal cut exceeds (roughly) half of their degree.

\subsection{Unified Algorithm for ConstrainedMinCut}
\label{subsec:unified}

The threshold $\OPT \ge \alpha n^2$ that separates the large-optimal-value and small-optimal-value cases is used only in the analysis.
The unified algorithm runs Algorithm~\ref{alg:largeOPT} and Algorithm~\ref{alg:small_case} in turn and returns whichever feasible cut is smaller.

\begin{algorithm}[H]
  \caption{ConstrainedMinCut$(G = (V, E), c, \rho, \zeta, \varepsilon)$}
  \label{alg:unified}
  \begin{algorithmic}[1]
    \State $(L_1, R_1)\gets \mathrm{ConstrainedMinCutLarge}(G, c, \rho, \zeta, \varepsilon)$ \Comment{Algorithm~\ref{alg:largeOPT}}
    \State $(L_2, R_2) \gets \mathrm{ConstrainedMinCutSmall}(G, c, \rho, \zeta, \varepsilon)$ \Comment{Algorithm~\ref{alg:small_case}}
    \State Let $\mathcal{F} \gets \set{L \in \set{L_1, L_2} \mid c(L) \in [\rho(1-\zeta), \rho(1+\zeta)]}$
    \State \Return $\arg\min_{L \in \mathcal{F}} \abs{\mE(L, \overline{L})}$
  \end{algorithmic}
\end{algorithm}

The randomized version of our main result, restated:
\eprasthm*

\begin{proof}[Proof of Theorem~\ref{thm:EPRAS}]
  Exactly one of $\OPT > \alpha n^2$ or $\OPT \le \alpha n^2$ holds.
  In the former case, by Theorem~\ref{thm:largeOPT}, $(L_1, R_1)$ is feasible and satisfies $\abs{\mE(L_1, \overline{L_1})} \le (1+\varepsilon)\OPT$ with probability at least $9/10$.
  In the latter case, by Theorem~\ref{thm:smallOPT}, $(L_2, R_2)$ is feasible and satisfies $\abs{\mE(L_2, \overline{L_2})} \le (1+\varepsilon)\OPT$ with high probability.
  In either case, $\mathcal{F}$ is nonempty with high probability, and the cut returned by Algorithm~\ref{alg:unified} is a $(1+\varepsilon)$-approximate feasible solution.

  Recall that $\alpha = \min\set{\frac{\varepsilon}{22}(1/2 - 3a)^2 \delta^2, 6 \rho' \zeta c_0\delta / 35}$.
  The running time is the sum of the running times of Algorithm~\ref{alg:largeOPT} and Algorithm~\ref{alg:small_case}, which is $n^{O(1)} + 2^{\tilde{O}(1/(\alpha\varepsilon)^2)} = n^{O(1)} + 2^{\tilde{O}(1/(\zeta^2\varepsilon^4))}$.
\end{proof}

\subsection{Derandomization (Proof of Theorem~\ref{thm:EPTAS})}
We can derandomize the above algorithm.
We use the randomized algorithm in two places:
\begin{enumerate}
  \item Obtain a partition of $V$ via the weak regularity lemma.
  \item Sample $S$ uniformly at random from $V$.
\end{enumerate}

Fox et al. showed that there is a deterministic algorithm to obtain a partition of $V$ with properties similar to those of the weak regularity lemma \cite{fox2019fast}.
\begin{thm}[Theorem 1.3 in \cite{fox2019fast}]
  \label{thm:deterministic_wrl}
  Suppose $A$ is an $n \times n$ matrix and suppose $\varepsilon > 0$.
  There exists a deterministic algorithm that runs in $\varepsilon^{-O(1)}n^2$ time, and finds a cut decomposition of $A$ of width $O(\varepsilon^{-16})$, coefficient length $O(1)$, and error at most $\varepsilon n^2$.
\end{thm}

We can derandomize the large-optimal-value case by replacing the randomized weak regularity lemma in Theorem~\ref{thm:weak_regularity} with Theorem~\ref{thm:deterministic_wrl}.
The increased width of the cut decomposition yields a running time of $2^{\tilde{O}(1/(\alpha \varepsilon)^{16})}$, which still falls within the EPTAS regime.

Next, we derandomize the sampling step.
Gillman showed a Chernoff bound for a random walk on a weighted graph in terms of the eigenvalue gap \cite{gillman1998chernoff}.
\begin{thm}[Theorem 2.1 in \cite{gillman1998chernoff}]
  \label{gillman}
  Let $G = (V, E)$ be a weighted graph with $n$ vertices.
  For $S \subseteq V$, let $t_S$ be the number of visits to $S$ in $t$ steps.
  Let $\lambda = 1 - \lambda_2$, where $\lambda_2$ is the second largest eigenvalue of the normalized adjacency matrix of $G$.
  Let $\pi$ be the stationary distribution of the random walk on $G$.
  Then, for all $u > 0$,
  \begin{equation*}
    \Pr\brac*{\abs*{\frac{t_S}{t} - \pi(S)} \ge u} \le 2(1 + u \lambda/10)\nor*{\frac{1}{\sqrt{\pi}}}_2 e^{-\frac{\lambda u^2 t}{20}}.
  \end{equation*}
\end{thm}
To apply Theorem~\ref{gillman}, we need an explicit constant-degree $d$-regular expander on exactly $n$ vertices with a constant spectral gap. 
Many explicit constructions of such expander families are known \cite{margulis1988explicit,lubotzky1986explicit,morgenstern1994existence,alon2021explicit}.
Fix an even constant $d \ge 4$. 
By Alon's construction of explicit expanders of every degree and size~\cite{alon2021explicit}, for every sufficiently large $n$, there is a deterministic polynomial-time construction of an $n$-vertex $d$-regular graph whose nontrivial adjacency eigenvalues are at most $2\sqrt{d-1}+\eta$ in absolute value, for any fixed constant $\eta>0$. 
For sufficiently small fixed $\eta$, this gives a constant spectral gap for the normalized adjacency matrix.
For derandomization, we enumerate all walks of length $t := k\log n$ on the expander over all choices of the initial vertex, where $k$ is a sufficiently large constant, and take the visited vertices of each walk as a multiset $S$.
Note that $S$ may contain duplicate vertices.
Since $\abs{N(v)\cap S} = \sum_{u \in S} \bbm[u \in N(v)]$ and $\abs{N(v) \cap S_L} = \sum_{u \in S} \bbm[u \in N(v) \cap \LOPT]$, we can derive the following concentration bounds:
\begin{align*}
  \Pr\brac*{\abs*{\frac{\abs{N(v) \cap S}}{\abs{S}} - \frac{d_v}{n}} \ge u}                     & \le 2(1 + u \lambda/10)\nor*{\frac{1}{\sqrt{\pi}}}_2 e^{-\frac{\lambda u^2 t}{20}}, \\
  \Pr\brac*{\abs*{\frac{\abs{N(v) \cap S_L}}{\abs{S}} - \frac{\abs{N(v) \cap \LOPT}}{n}} \ge u} & \le 2(1 + u \lambda/10)\nor*{\frac{1}{\sqrt{\pi}}}_2 e^{-\frac{\lambda u^2 t}{20}}.
\end{align*}
When $\abs*{\tfrac{\abs{N(v) \cap S}}{\abs{S}} - \tfrac{d_v}{n}} \le u$ and
$\abs*{\tfrac{\abs{N(v) \cap S_L}}{\abs{S}} - \tfrac{\abs{N(v) \cap \LOPT}}{n}} \le u$ hold, we can bound
\begin{align*}
  \abs*{\hat{p}_L(v) - p_L(v)}
   & = \abs*{\frac{\abs{N(v)\cap S_L}/\abs{S}}{\abs{N(v)\cap S}/\abs{S}} - \frac{\abs{N(v)\cap \LOPT}/n}{d_v/n}}                                 \\
   & = \bigg \lvert\frac{\abs{N(v)\cap S_L}/\abs{S}}{\abs{N(v)\cap S}/\abs{S}} - \frac{\abs{N(v)\cap \LOPT}/n}{\abs{N(v)\cap S}/\abs{S}}         \\
   & \quad + \frac{\abs{N(v)\cap \LOPT}/n}{\abs{N(v)\cap S}/\abs{S}} - \frac{\abs{N(v)\cap \LOPT}/n}{d_v/n} \bigg \rvert                         \\
   & \le \abs*{\frac{\abs{N(v)\cap S_L}/\abs{S}}{\abs{N(v)\cap S}/\abs{S}} - \frac{\abs{N(v)\cap \LOPT}/n}{\abs{N(v)\cap S}/\abs{S}}}            \\
   & \quad + \abs*{\frac{\abs{N(v)\cap \LOPT}/n}{\abs{N(v)\cap S}/\abs{S}} - \frac{\abs{N(v)\cap \LOPT}/n}{d_v/n}}                               \\
   & = \frac{\abs{\abs{N(v) \cap S_L}/\abs{S} - \abs{N(v) \cap \LOPT}/n}}{\abs{N(v) \cap S}/\abs{S}}                                             \\
   & \quad + \frac{\abs{N(v) \cap \LOPT}/n}{d_v/n} \cdot {\frac{\abs{\abs{N(v)\cap S}/\abs{S} - d_v/n}}{\abs{N(v)\cap S}/\abs{S}}}               \\
   & \le \frac{2u}{\abs{N(v) \cap S}/\abs{S}}                                                                                                    \\
   & \le \frac{2u}{\delta - u}. \qquad (\because d_v/n \ge \delta \text{ and } \abs*{\tfrac{\abs{N(v) \cap S}}{\abs{S}} - \tfrac{d_v}{n}} \le u)
\end{align*}
If we set $u = \tfrac{a\delta}{2 + a}$, we have $\abs{\hat{p}_L(v) - p_L(v)} \le a$ with probability at least $1 - Cn^{1 - \frac{\lambda u^2 k}{20}}$ for some constant $C > 0$.
By union bound, the above inequality holds for all $v \in V$ with probability at least $1 - Cn^{2 - \frac{\lambda u^2 k}{20}}$.
The number of all walks of length $O(\log{n})$ for a constant-degree expander graph is $n \cdot d^{O(\log{n})} = n^{O(1)}$.
Thus, by enumerating all such walks and running the subsequent classification and adjustment procedure on each, we are guaranteed that at least one of these walks satisfies the concentration bound $|\hat{p}_L(v) - p_L(v)| \le a$, ensuring that the minimum-sized valid cut we output across all iterations is a $(1+\varepsilon)$-approximate solution.
This completes the derandomization of the previous algorithm and the proof of Theorem~\ref{thm:EPTAS}, which in particular establishes the main theorem (Theorem~\ref{thm:eptas_con}).

Note that the same approach, applying the weak regularity lemma and derandomization, also works for \textsc{BalancedSeparator} as in \cite{ARORA1999193, frieze1996regularity}.
Therefore, \textsc{BalancedSeparator} admits an EPTAS on everywhere-$\delta$-dense graphs.

\section{MinQuotientCut and ProductSparsestCut}
\label{sec:red}

In this section, we present EPTASs for the \textsc{MinQuotientCut} problem and the \textsc{ProductSparsestCut} problem through a short reduction from the EPTAS for \textsc{ConstrainedMinCut}.
\reduction*
\corollaryred*

First, we confirm that the above four problems are special cases of \textsc{MinQuotientCut} and \textsc{ProductSparsestCut}.

For \textsc{EdgeExpansion} (resp. \textsc{UniformSparsestCut}), we set $c_v = n$ for every vertex $v \in V$, and then $c(S) = n\abs{S}$.
Similarly, for \textsc{Conductance} (resp. \textsc{NormalizedCut}), we set $c_v = d_v$ for every vertex $v \in V$, and then $c(S) = \vol(S)$.
In this setting, the form of the objective is the same except for a scaling factor of $n$ or $n^2$.

Below, we prove Theorem~\ref{thm:mqc_psc}.

\subsection{MinQuotientCut}
\subsubsection{When the Optimal Cut Is Unbalanced}
\begin{lem}
  Let $G$ be an everywhere-$\delta$-dense graph.
  For all cuts $\mE(S, \overline{S})$ with $\abs{S} < \varepsilon \delta n$, the quotient cut $q(S)$ can be approximated by $\tfrac{\vol(S)}{c(S)}$ within $1 \pm \varepsilon$ multiplicative error and $\tfrac{\vol(S)}{c(S)}$ can be minimized in $O(n)$ time.
\end{lem}

\begin{proof}
  Since $\abs{S} < \varepsilon \delta n$, for sufficiently small $\varepsilon$, $c(S) \le \abs{S}n \le \varepsilon \delta n^2 < (1-\varepsilon \delta) c_0 n^2 \le c(\overline{S})$ holds.
  Hence, we can assume that $\min\set{c(S), c(\overline{S})} = c(S)$.

  We have:
  \begin{align*}
    (1-\varepsilon)\frac{\vol(S)}{c(S)}\le \frac{\vol(S) - \abs{S}^2}{c(S)}\le q(S) = \frac{\abs{\mE(S, \overline{S})}}{c(S)} & \le \frac{\vol(S)}{c(S)}. \\
  \end{align*}
  Thus, for all cuts $\mE(S, \overline{S})$ with $\abs{S} < \varepsilon \delta n$,
  \begin{equation*}
    (1-\varepsilon)\frac{\vol(S)}{c(S)} \le q(S) \le (1+\varepsilon)\frac{\vol(S)}{c(S)}.
  \end{equation*}
  Since $\min_{\emptyset \ne S \subseteq V} \tfrac{\vol(S)}{c(S)} \ge \min_{v \in V} \tfrac{\vol(\set{v})}{c_v}$ holds, we can minimize $\tfrac{\vol(S)}{c(S)}$ by choosing a single vertex that minimizes $\tfrac{\vol(\set{v})}{c_v}$.
\end{proof}
Thus, when the optimal cut $S^*$ has $\abs{S^*} < \varepsilon \delta n$, we obtain a $(1 \pm \varepsilon)$-approximation of $q(S^*)$ in $O(n)$ time.

\subsubsection{When the Optimal Cut Is Balanced}

We now consider the case where $\abs{S^*}$ is large.
We apply an EPTAS for \textsc{ConstrainedMinCut} as a subroutine to approximate the \textsc{MinQuotientCut} problem.
The value of $c(S^*)$ for the optimal balanced cut lies in $\set{\ceil{c_0\varepsilon \delta n^2}, \ldots, \floor{\tfrac{c(V)}{2}}}$ for \textsc{MinQuotientCut}.
Thus, we divide possible values of $c(S^*)$ into segments, run the EPTAS for \textsc{ConstrainedMinCut} for each parameter setting, and return the best solution among them.
Let $S^*$ denote the optimal cut for the \textsc{MinQuotientCut} problem.

\begin{lem}
  \label{lem:parameter_setting_mqc}
  For a given $\varepsilon > 0$, there exist a constant $\zeta$ and a sequence of $\rho_0, \rho_1, \ldots$ such that the following hold:
  \begin{itemize}
    \item $\rho_0(1-\zeta) = c_0\varepsilon \delta n^2$.
    \item $\rho_i(1+\zeta) = \rho_{i+1}(1 - \zeta)$.
    \item The minimum index $k$ such that $\rho_k(1+\zeta) \ge \tfrac{c(V)}{2}$ is at most $O\pare*{\frac{1}{\varepsilon}\log{\frac{1}{\varepsilon}}}$.
    \item $\tfrac{1+\zeta}{1-\zeta} = 1 + \varepsilon$.
  \end{itemize}
\end{lem}

\begin{proof}
  For a given $\varepsilon > 0$, we fix $\zeta = \tfrac{\varepsilon}{2 + \varepsilon}$, which satisfies $\tfrac{1 + \zeta}{1 - \zeta} = 1 + \varepsilon$,
  and we define $\rho_i$ as $\left(\tfrac{1 + \zeta}{1 - \zeta}\right)^i\tfrac{c_0\varepsilon \delta n^2}{1 - \zeta} = (1 + \varepsilon)^i \cdot \tfrac{c_0\varepsilon \delta n^2}{1 - \zeta}$.

  $\rho_0(1-\zeta) = c_0\varepsilon \delta n^2$ and $\rho_{i+1}(1-\zeta) = \rho_i(1+\zeta)$ hold by definition.
  The minimum index $k$ that satisfies $\rho_{k}(1+\zeta) \ge \tfrac{c(V)}{2}$ can be obtained as follows:
  \begin{align*}
    \rho_k(1+\zeta)
          & = (1 + \varepsilon)^{k+1} c_0\varepsilon \delta n^2  \ge n^2/2 \ge c(V)/2                                                               \\
    (1 + \varepsilon)^{k+1}
          & \ge \frac{1}{2c_0\varepsilon \delta}                                                                                                    \\
    k + 1 & \ge \frac{\log{\frac{1}{2c_0\varepsilon \delta}}}{\log{(1 + \varepsilon)}} = O\pare*{\frac{1}{\varepsilon}\log{\frac{1}{\varepsilon}}}.
  \end{align*}
  This completes the proof.
\end{proof}

\begin{rem}
  \label{rem:clamp_mqc}
  By the definition of $k$, $\rho_k(1+\zeta)$ may exceed $c(V)/2$.
  Since the optimal cut $S^*$ for \textsc{MinQuotientCut} satisfies $\min\set{c(S^*), c(\overline{S^*})} \le c(V)/2$, we may assume without loss of generality that $c(S^*) \le c(V)/2$.
  Therefore, for the largest index $k$, we replace $\rho_k$ with the clamped value $\rho_k' := \tfrac{c(V)}{2(1+\zeta)}$ so that the cost interval at index $k$ becomes $[\rho_k'(1-\zeta), c(V)/2]$.
  This clamp does not create a gap in the coverage of $[\rho_0(1-\zeta), c(V)/2]$ because $\rho_k'(1-\zeta) = \tfrac{c(V)(1-\zeta)}{2(1+\zeta)} \le \rho_k(1-\zeta)$ holds by $\rho_k(1+\zeta) \ge c(V)/2$.
  Since $\zeta$ is unchanged, the runtime of the subroutine is not worsened.

  Consequently, there always exists $j \in \set{0, \ldots, k}$ such that $c(S^*) \in [\rho_j(1-\zeta), \rho_j(1+\zeta)]$ and $\rho_j(1+\zeta) \le c(V)/2$, which justifies the bound $\min\set{c(S^*), c(\overline{S^*})} \le \rho_j(1+\zeta)$ in the analysis below.
\end{rem}

We run \textsc{ConstrainedMinCut}$(G, c, \rho_j, \zeta, \varepsilon/3)$ with $\rho_j, \zeta$ from Lemma~\ref{lem:parameter_setting_mqc} (with the clamp at $j = k$ in Remark~\ref{rem:clamp_mqc}) for $j \in \set{0, \ldots, k}$, and denote the output cut by $S_j$.
We take the best cut among them, which is shown to be a $(1 + \varepsilon)$-approximate solution for \textsc{MinQuotientCut} as follows:
\begin{align*}
  (1 + \varepsilon/3)q(S^*)
   & = (1 + \varepsilon/3)\frac{\abs{\mE(S^*, \overline{S^*})}}{\min\set{c(S^*), c(\overline{S^*})}} \\
   & \ge \min_{0 \le j \le k} \frac{\abs{\mE(S_j, \overline{S_j})}}{(1+\zeta)\rho_j}                 \\
   & \ge \frac{1-\zeta}{1+\zeta}\frac{\abs{\mE(S_{j^*}, \overline{S_{j^*}})}}{c(S_{j^*})}            \\
   & = \frac{1}{1+\varepsilon/3}q(S_{j^*}).
\end{align*}
Here, $j^*$ is an index of the best cut among $S_j$.
The first inequality follows by running the EPTAS for \textsc{ConstrainedMinCut} with parameters $\rho_j, \zeta$ for an index $j$ such that $c(S^*) \in [\rho_j(1-\zeta), \rho_j(1 + \zeta)]$ with precision parameter $\varepsilon/3$.
Its running time is $\varepsilon^{-O(1)}n^{O(1)} + 2^{\tilde{O}(1/\varepsilon^{48})}$, since we run the EPTAS for \textsc{ConstrainedMinCut} $\tilde{O}(1/\varepsilon)$ times, each with $\zeta = \Theta(\varepsilon)$ and precision parameter $\varepsilon/3$.

\subsection{ProductSparsestCut}
\subsubsection{When the Optimal Cut Is Unbalanced}
\begin{lem}
  Let $G$ be an everywhere-$\delta$-dense graph.
  For all cuts $\mE(S, \overline{S})$ with $\abs{S} \le \min\set{\varepsilon\delta n, \tfrac{\varepsilon}{2}c_0n}$, the sparsity $\Phi^\times(S)$ can be approximated by $\tfrac{\vol(S)}{c(V)c(S)}$ within $1 \pm \varepsilon$ multiplicative error and $\tfrac{\vol(S)}{c(V)c(S)}$ can be minimized in $O(n)$ time.
\end{lem}
\begin{proof}
  We have:
  \begin{align*}
    (1-\varepsilon)\frac{\vol(S)}{c(V)c(S)}\le \frac{\vol(S) - \abs{S}^2}{c(S)(c(V) - c(S))} \le \Phi^\times(S) = \frac{\abs{\mE(S, \overline{S})}}{c(S)c(\overline{S})} & \le \frac{\vol(S)}{c(S)c(V)(1 - c(S)/c(V))} \le \frac{\vol(S)}{c(S)c(V)} (1 + \varepsilon). \\
  \end{align*}
  Thus, for all cuts $\mE(S, \overline{S})$ with $\abs{S} \le \min\set{\varepsilon\delta n, \tfrac{\varepsilon}{2}c_0n}$, we have
  \begin{equation*}
    (1-\varepsilon)\frac{\vol(S)}{c(V)c(S)} \le \Phi^\times(S) \le (1+\varepsilon)\frac{\vol(S)}{c(V)c(S)}.
  \end{equation*}
\end{proof}
As in the previous section, $\frac{\vol(S)}{c(V)c(S)}$ can be minimized in $O(n)$ time.

\subsubsection{When the Optimal Cut Is Balanced}
\begin{lem}
  \label{lem:parameter_setting_psc}
  For a given $\varepsilon > 0$, there exist a constant $\zeta$ and a sequence of $\rho_0, \rho_1, \ldots$ such that the following hold:
  \begin{itemize}
    \item $\rho_0(1-\zeta) \le \min\set{\varepsilon\delta c_0 n^2, \tfrac{\varepsilon}{2}c_0^2 n^2}$.
    \item $\rho_{i+1}(1-\zeta) = \rho_i(1+\zeta)$.
    \item The minimum index $k$ such that $\rho_{k}(1+\zeta) \ge c(V)/2$ is at most $O\pare*{\frac{1}{\varepsilon}\log{\frac{1}{\varepsilon}}}$.
    \item For $^\forall i \in \set{0, 1, \ldots, k}$, $\tfrac{1 + \zeta}{1-\zeta}\tfrac{c(V) -\rho_i(1 -\zeta)}{c(V) - \rho_i(1 + \zeta)} \le 1 + \varepsilon$.
  \end{itemize}
\end{lem}
\begin{proof}
  We fix $\zeta = \tfrac{\varepsilon/2}{2 + 3\varepsilon/2}$ and define $\rho_i$ as $\pare*{\tfrac{1 + \zeta}{1 - \zeta}}^i\tfrac{\min\set{\varepsilon\delta c_0 n^2, \tfrac{\varepsilon}{2}c_0^2 n^2}}{1 - \zeta}$.
  $\rho_0(1-\zeta) = \min\set{\varepsilon\delta c_0 n^2, \tfrac{\varepsilon}{2}c_0^2 n^2}$ and $\rho_{i+1}(1-\zeta) = \rho_i(1+\zeta)$ hold by definition.
  The minimum index $k$ that satisfies $\rho_{k}(1+\zeta) \ge c(V)/2$ is obtained as follows:
  \begin{equation*}
    \rho_{k}(1+\zeta)
    = (1 + \frac{\varepsilon}{2 +\varepsilon})^{k+1} \cdot \min\set{\varepsilon\delta c_0 n^2, \tfrac{\varepsilon}{2}c_0^2 n^2} \ge n^2 \ge c(V)/2
  \end{equation*}
  \begin{equation*}
    (1 + \frac{\varepsilon}{2 +\varepsilon})^{k+1}
    \ge \frac{2}{\varepsilon c_0 \min\set{\delta, c_0}}
  \end{equation*}
  \begin{equation*}
    k + 1 \ge \frac{\log{\frac{2}{\varepsilon c_0 \min\set{\delta, c_0}}}}{\log{(1 + \varepsilon/(2 + \varepsilon))}} = O\pare*{\frac{1}{\varepsilon}\log{\frac{1}{\varepsilon}}}.
  \end{equation*}

  Finally, we prove the last condition.
  We can assume $\rho_i(1 - \zeta) < c(V)/2$ for all $i \in \set{0, 1, \ldots, k}$ because $\rho_{k-1}(1 + \zeta) = \rho_k(1 - \zeta) < c(V)/2$ holds by the definition of $k$.
  The last condition is equivalent to the following:
  \begin{equation*}
    2\zeta c(V) + \varepsilon \rho_i (1 -\zeta^2) \le \varepsilon(1-\zeta)c(V).
  \end{equation*}
  Since $\rho_i(1-\zeta)<c(V)/2$, we have
  \begin{equation*}
    \varepsilon\rho_i(1-\zeta^2)
    = \varepsilon\rho_i(1-\zeta)(1+\zeta)
    < \frac{\varepsilon(1+\zeta)}{2}c(V).
  \end{equation*}
  Thus, it suffices to show
  \begin{equation*}
    2\zeta c(V)+\frac{\varepsilon(1+\zeta)}{2}c(V) \le \varepsilon(1-\zeta)c(V).
  \end{equation*}
  This holds by the choice
  \begin{equation*}
    \zeta=\frac{\varepsilon/2}{2+3\varepsilon/2}.
  \end{equation*}

  This completes the proof.
\end{proof}

\begin{rem}
  \label{rem:clamp_psc}
  As in Remark~\ref{rem:clamp_mqc}, when $\rho_k > c(V)/2$ we replace $\rho_k$ with the clamped value $\rho_k' := \tfrac{c(V)}{2(1+\zeta)}$ so that the input $\rho$ to \textsc{ConstrainedMinCut} satisfies $\rho \le c(V)/2$ as required.
\end{rem}

We run \textsc{ConstrainedMinCut}$(G, c, \rho_j, \zeta, \varepsilon/3)$ with $\rho_j, \zeta$ from Lemma~\ref{lem:parameter_setting_psc} (instantiated with $\varepsilon/3$ in place of $\varepsilon$, and with the clamp at $j = k$ in Remark~\ref{rem:clamp_psc}) for $j \in \set{0, \ldots, k}$, and denote the output cut by $S_j$.
We take the best cut among them, which is shown to be a $(1 + \varepsilon)$-approximate solution for \textsc{ProductSparsestCut} as follows:
\begin{align*}
  (1 + \varepsilon/3)\Phi^\times(S^*)
   & = (1 + \varepsilon/3)\frac{\abs{\mE(S^*, \overline{S^*})}}{c(S^*)c(\overline{S^*})}                                                                                    \\
   & \ge \min_{0 \le j \le k} \frac{\abs{\mE(S_j, \overline{S_j})}}{\rho_j(1 + \zeta)(c(V) - \rho_j(1 - \zeta))}                                                            \\
   & \ge \min_{0 \le j \le k} \frac{1-\zeta}{1 + \zeta}\frac{c(V) - \rho_j(1+\zeta)}{c(V) - \rho_j(1-\zeta)} \frac{\abs{\mE(S_j, \overline{S_j})}}{c(S_j)c(\overline{S_j})} \\
   & \ge \frac{1}{1 + \varepsilon/3} \Phi^\times(S_{j^*}).
\end{align*}
Here, $j^*$ is an index of the best cut among $S_j$.

Thus, we can find a cut $\mE(S, \overline{S})$ such that $\Phi^\times(S) \le (1+\varepsilon/3)^2\Phi^\times(S^*) \le (1 + \varepsilon)\Phi^\times(S^*)$ (we can assume $\varepsilon < 1$ without loss of generality).
Its running time is $\varepsilon^{-O(1)}n^{O(1)} + 2^{\tilde{O}(1/\varepsilon^{48})}$, as in the previous section.

\begin{algorithm}[H]
  \caption{\textsc{MinQuotientCut}$(G, c, \varepsilon)$}
  \label{alg:cond}
  \begin{algorithmic}[1]
    \State $\zeta, \set{\rho_0, \ldots, \rho_k}$ $\gets$ parameters in Lemma~\ref{lem:parameter_setting_mqc} with error parameter $\varepsilon/3$
    \State $S$ $\gets$ the minimizer for $\vol(\set{v})/c(\set{v})$ for $v \in V$
    \For {$j$ in $0$..$k$}
    \State $S_j \gets$ \textsc{ConstrainedMinCut}$(G, c, \rho_j, \zeta, \varepsilon/3)$
    \If {$q(S_j) < q(S)$}
    \State $S \gets S_j$
    \EndIf
    \EndFor
    \State \Return $S$
  \end{algorithmic}
\end{algorithm}

\begin{algorithm}[H]
  \caption{\textsc{ProductSparsestCut}$(G, c, \varepsilon)$}
  \label{alg:ncut}
  \begin{algorithmic}[1]
    \State $\zeta, \set{\rho_0, \ldots, \rho_k}$ $\gets$ parameters in Lemma~\ref{lem:parameter_setting_psc} with error parameter $\varepsilon/3$
    \State $S$ $\gets$ the minimizer for $\vol(\set{v})/c(\set{v})$ for $v \in V$
    \For {$j$ in $0$..$k$}
    \State $S_j \gets$ \textsc{ConstrainedMinCut}$(G, c, \rho_j, \zeta, \varepsilon/3)$
    \If {$\Phi^\times(S_j) < \Phi^\times(S)$}
    \State $S \gets S_j$
    \EndIf
    \EndFor
    \State \Return $S$
  \end{algorithmic}
\end{algorithm}

\bibliography{ref}

@article{frieze1999quick,
  title     = {Quick approximation to matrices and applications},
  author    = {Frieze, Alan and Kannan, Ravi},
  journal   = {Combinatorica},
  volume    = {19},
  number    = {2},
  pages     = {175--220},
  year      = {1999},
  publisher = {Springer}
}

@article{ARORA1999193,
  title    = {Polynomial Time Approximation Schemes for Dense Instances of NP-Hard Problems},
  journal  = {Journal of Computer and System Sciences},
  volume   = {58},
  number   = {1},
  pages    = {193-210},
  year     = {1999},
  issn     = {0022-0000},
  doi      = {https://doi.org/10.1006/jcss.1998.1605},
  url      = {https://www.sciencedirect.com/science/article/pii/S0022000098916051},
  author   = {Sanjeev Arora and David Karger and Marek Karpinski}
}

@inproceedings{guruswami2011lasserre,
  title        = {Lasserre hierarchy, higher eigenvalues, and approximation schemes for graph partitioning and quadratic integer programming with PSD objectives},
  author       = {Guruswami, Venkatesan and Sinop, Ali Kemal},
  booktitle    = {2011 IEEE 52nd Annual Symposium on Foundations of Computer Science},
  pages        = {482--491},
  year         = {2011},
  organization = {IEEE}
}

@inproceedings{oveis2013new,
  title        = {A new regularity lemma and faster approximation algorithms for low threshold rank graphs},
  author       = {Oveis Gharan, Shayan and Trevisan, Luca},
  booktitle    = {International Workshop on Approximation Algorithms for Combinatorial Optimization},
  pages        = {303--316},
  year         = {2013},
  organization = {Springer}
}

@inproceedings{yoshida2014approximation,
  title     = {Approximation schemes via Sherali-Adams hierarchy for dense constraint satisfaction problems and assignment problems},
  author    = {Yoshida, Yuichi and Zhou, Yuan},
  booktitle = {Proceedings of the 5th conference on Innovations in theoretical computer science},
  pages     = {423--438},
  year      = {2014}
}

@article{arora1998proof,
  title     = {Proof verification and the hardness of approximation problems},
  author    = {Arora, Sanjeev and Lund, Carsten and Motwani, Rajeev and Sudan, Madhu and Szegedy, Mario},
  journal   = {Journal of the ACM (JACM)},
  volume    = {45},
  number    = {3},
  pages     = {501--555},
  year      = {1998},
  publisher = {ACM New York, NY, USA}
}

@inproceedings{frieze1996regularity,
  title        = {The regularity lemma and approximation schemes for dense problems},
  author       = {Frieze, Alan and Kannan, Ravi},
  booktitle    = {Proceedings of 37th conference on foundations of computer science},
  pages        = {12--20},
  year         = {1996},
  organization = {IEEE}
}

@article{goldreich1998property,
  title     = {Property testing and its connection to learning and approximation},
  author    = {Goldreich, Oded and Goldwasser, Shari and Ron, Dana},
  journal   = {Journal of the ACM (JACM)},
  volume    = {45},
  number    = {4},
  pages     = {653--750},
  year      = {1998},
  publisher = {ACM New York, NY, USA}
}

@inproceedings{mathieu2008yet,
  title     = {Yet another algorithm for dense max cut: go greedy.},
  author    = {Mathieu, Claire and Schudy, Warren},
  booktitle = {SODA},
  volume    = {8},
  pages     = {176--182},
  year      = {2008}
}

@book{de2000polynomial,
  title     = {Polynomial time approximation of dense weighted instances of MAX-CUT},
  author    = {de la Vega, Wenceslas Fernandez and Karpinski, Marek},
  year      = {2000},
  publisher = {Inst. f{\"u}r Informatik}
}

@article{karpinski2001polynomial,
  title     = {Polynomial time approximation schemes for some dense instances of NP-hard optimization problems},
  author    = {Karpinski, Marek},
  journal   = {Algorithmica},
  volume    = {30},
  number    = {3},
  pages     = {386--397},
  year      = {2001},
  publisher = {Springer}
}

@article{fox2019fast,
  title     = {A fast new algorithm for weak graph regularity},
  author    = {Fox, Jacob and Lov{\'a}sz, L{\'a}szl{\'o} Mikl{\'o}s and Zhao, Yufei},
  journal   = {Combinatorics, Probability and Computing},
  volume    = {28},
  number    = {5},
  pages     = {777--790},
  year      = {2019},
  publisher = {Cambridge University Press}
}

@article{gillman1998chernoff,
  title     = {A Chernoff bound for random walks on expander graphs},
  author    = {Gillman, David},
  journal   = {SIAM Journal on Computing},
  volume    = {27},
  number    = {4},
  pages     = {1203--1220},
  year      = {1998},
  publisher = {SIAM}
}

@inproceedings{raghavendra2010graph,
  title     = {Graph expansion and the unique games conjecture},
  author    = {Raghavendra, Prasad and Steurer, David},
  booktitle = {Proceedings of the forty-second ACM symposium on Theory of computing},
  pages     = {755--764},
  year      = {2010}
}

@inproceedings{austrin2012inapproximability,
  title        = {Inapproximability of treewidth, one-shot pebbling, and related layout problems},
  author       = {Austrin, Per and Pitassi, Toniann and Wu, Yu},
  booktitle    = {International Workshop on Approximation Algorithms for Combinatorial Optimization},
  pages        = {13--24},
  year         = {2012},
  organization = {Springer}
}

@article{manurangsi2018inapproximability,
  title     = {Inapproximability of maximum biclique problems, minimum k-cut and densest at-least-k-subgraph from the small set expansion hypothesis},
  author    = {Manurangsi, Pasin},
  journal   = {Algorithms},
  volume    = {11},
  number    = {1},
  pages     = {10},
  year      = {2018},
  publisher = {MDPI}
}

@inproceedings{gandhi2014set,
  title        = {On set expansion problems and the small set expansion conjecture},
  author       = {Gandhi, Rajiv and Kortsarz, Guy},
  booktitle    = {International Workshop on Graph-Theoretic Concepts in Computer Science},
  pages        = {189--200},
  year         = {2014},
  organization = {Springer}
}

@article{khot2007optimal,
  title     = {Optimal inapproximability results for MAX-CUT and other 2-variable CSPs?},
  author    = {Khot, Subhash and Kindler, Guy and Mossel, Elchanan and O’Donnell, Ryan},
  journal   = {SIAM Journal on Computing},
  volume    = {37},
  number    = {1},
  pages     = {319--357},
  year      = {2007},
  publisher = {SIAM}
}

@article{khot2008vertex,
  title     = {Vertex cover might be hard to approximate to within 2- $\varepsilon$},
  author    = {Khot, Subhash and Regev, Oded},
  journal   = {Journal of Computer and System Sciences},
  volume    = {74},
  number    = {3},
  pages     = {335--349},
  year      = {2008},
  publisher = {Elsevier}
}

@inproceedings{raghavendra2008optimal,
  title     = {Optimal algorithms and inapproximability results for every CSP?},
  author    = {Raghavendra, Prasad},
  booktitle = {Proceedings of the fortieth annual ACM symposium on Theory of computing},
  pages     = {245--254},
  year      = {2008}
}

@inproceedings{karpinski2009linear,
  title     = {Linear time approximation schemes for the Gale-Berlekamp game and related minimization problems},
  author    = {Karpinski, Marek and Schudy, Warren},
  booktitle = {Proceedings of the forty-first annual ACM symposium on Theory of computing},
  pages     = {313--322},
  year      = {2009}
}

@article{bazgan2003polynomial,
  title     = {Polynomial time approximation schemes for dense instances of minimum constraint satisfaction},
  author    = {Bazgan, Cristina and Fernandez de La Vega, W and Karpinski, Marek},
  journal   = {Random Structures \& Algorithms},
  volume    = {23},
  number    = {1},
  pages     = {73--91},
  year      = {2003},
  publisher = {Wiley Online Library}
}

@inproceedings{bazgan2002approximability,
  title        = {Approximability of dense instances of NEAREST CODEWORD problem},
  author       = {Bazgan, Cristina and de la Vega, W Fernandez and Karpinski, Marek},
  booktitle    = {Scandinavian Workshop on Algorithm Theory},
  pages        = {298--307},
  year         = {2002},
  organization = {Springer}
}

@article{giotis2005correlation,
  title   = {Correlation clustering with a fixed number of clusters},
  author  = {Giotis, Ioannis and Guruswami, Venkatesan},
  journal = {arXiv preprint cs/0504023},
  year    = {2005}
}

@article{margulis1988explicit,
  title={Explicit group-theoretical constructions of combinatorial schemes and their application to the design of expanders and concentrators},
  author={Margulis, Grigorii Aleksandrovich},
  journal={Problemy peredachi informatsii},
  volume={24},
  number={1},
  pages={51--60},
  year={1988},
  publisher={Russian Academy of Sciences, Branch of Informatics, Computer Equipment and~…}
}

@inproceedings{lubotzky1986explicit,
  title={Explicit expanders and the Ramanujan conjectures},
  author={Lubotzky, Alexander and Phillips, Ralph and Sarnak, Peter},
  booktitle={Proceedings of the eighteenth annual ACM symposium on Theory of computing},
  pages={240--246},
  year={1986}
}

@article{morgenstern1994existence,
  title={Existence and explicit constructions of q+ 1 regular Ramanujan graphs for every prime power q},
  author={Morgenstern, Moshe},
  journal={Journal of Combinatorial Theory, Series B},
  volume={62},
  number={1},
  pages={44--62},
  year={1994},
  publisher={Elsevier}
}

@article{alon2021explicit,
  title={Explicit expanders of every degree and size},
  author={Alon, Noga},
  journal={Combinatorica},
  volume={41},
  number={4},
  pages={447--463},
  year={2021},
  publisher={Springer}
}
\bibliographystyle{plain}
\end{document}